\documentclass[draftclsnofoot,onecolumn,12pt]{IEEEtran}

\usepackage[T1]{fontenc}
\usepackage{amsmath,amssymb,amsfonts,mathrsfs,bm}
\usepackage{mathtools}
\usepackage{amsthm}
\usepackage{cite}
\usepackage[shortlabels]{enumitem}
\usepackage{graphicx}
\usepackage{epstopdf}
\usepackage{url}
\usepackage{colortbl}
\usepackage{multirow}
\usepackage{xcolor}
\usepackage[normalem]{ulem}
\usepackage{array}
\newcolumntype{L}[1]{>{\raggedright\let\newline\\\arraybackslash\hspace{0pt}}m{#1}}
\newcolumntype{C}[1]{>{\centering\let\newline\\\arraybackslash\hspace{0pt}}m{#1}}
\newcolumntype{R}[1]{>{\raggedleft\let\newline\\\arraybackslash\hspace{0pt}}m{#1}}

\makeatletter
\let\MYcaption\@makecaption
\makeatother
\usepackage[font=footnotesize]{subcaption}
\makeatletter
\let\@makecaption\MYcaption
\makeatother

\usepackage{xparse}



\usepackage{glossaries}
\newacronym{wrt}{w.r.t.}{with respect to}
\newacronym{RHS}{R.H.S.}{right-hand side}
\newacronym{LHS}{L.H.S.}{left-hand side}
\newacronym{iid}{i.i.d.}{independent and identically distributed}

\usepackage{float}

\ifx\notloadhyperref\undefined
	\ifx\loadbibentry\undefined
		\usepackage[hidelinks,hypertexnames=false]{hyperref} 
	\else
		\usepackage{bibentry}
		\makeatletter\let\saved@bibitem\@bibitem\makeatother
		\usepackage[hidelinks,hypertexnames=false]{hyperref}
		\makeatletter\let\@bibitem\saved@bibitem\makeatother
	\fi
\else
	\relax
\fi

\usepackage[capitalize]{cleveref}
\crefname{equation}{}{}
\Crefname{equation}{}{}
\crefname{claim}{claim}{claims}
\crefname{step}{step}{steps}
\crefname{line}{line}{lines}
\crefname{condition}{condition}{conditions}
\crefname{dmath}{}{}
\crefname{dseries}{}{}
\crefname{dgroup}{}{}
\crefname{Theorem}{Theorem}{Theorems}
\crefname{Corollary}{Corollary}{Corollaries}
\crefname{Proposition}{Proposition}{Propositions}
\crefname{Lemma}{Lemma}{Lemmas}
\crefname{Definition}{Definition}{Definitions}
\crefname{Example}{Example}{Examples}
\crefname{Assumption}{Assumption}{Assumptions}
\crefname{Remark}{Remark}{Remarks}
\crefname{Rem}{Remark}{Remarks}
\crefname{remarks}{Remarks}{Remarks}
\crefname{Exercise}{Exercise}{Exercises}
\crefname{Theorem_A}{Theorem}{Theorems}
\crefname{Corollary_A}{Corollary}{Corollaries}
\crefname{Proposition_A}{Proposition}{Propositions}
\crefname{Lemma_A}{Lemma}{Lemmas}
\crefname{Definition_A}{Definition}{Definitions}

\usepackage{algorithm,algorithmic}

\ifx\loadbreqn\undefined
	\relax
\else
	\usepackage{breqn} 
\fi

\ifx\renewtheorem\undefined
\ifx\useTheoremCounter\undefined
\newtheorem{Theorem}{Theorem}
\newtheorem{Corollary}{Corollary}
\newtheorem{Proposition}{Proposition}
\newtheorem{Lemma}{Lemma}
\else

\newtheorem{Proposition}[theorem]{Proposition}
\fi

\newtheorem{Example}{Example}

\fi

\theoremstyle{remark}

\theoremstyle{plain}

\newcommand{\Real}{\mathbb{R}}

\newcommand{\calF}{\mathcal{F}}

\newcommand{\calI}{\mathcal{I}}
\newcommand{\calJ}{\mathcal{J}}

\newcommand{\calP}{\mathcal{P}}
\newcommand{\calQ}{\mathcal{Q}}

\newcommand{\ba}{\mathbf{a}}
\newcommand{\bA}{\mathbf{A}}
\newcommand{\bb}{\mathbf{b}}
\newcommand{\bB}{\mathbf{B}}

\newcommand{\bC}{\mathbf{C}}

\newcommand{\bD}{\mathbf{D}}

\newcommand{\bF}{\mathbf{F}}

\newcommand{\bG}{\mathbf{G}}

\newcommand{\bH}{\mathbf{H}}

\newcommand{\bI}{\mathbf{I}}

\newcommand{\bn}{\mathbf{n}}

\newcommand{\bp}{\mathbf{p}}
\newcommand{\bP}{\mathbf{P}}

\newcommand{\bQ}{\mathbf{Q}}

\newcommand{\bR}{\mathbf{R}}

\newcommand{\bT}{\mathbf{T}}

\newcommand{\bU}{\mathbf{U}}
\newcommand{\bv}{\mathbf{v}}
\newcommand{\bV}{\mathbf{V}}

\newcommand{\bx}{\mathbf{x}}

\newcommand{\by}{\mathbf{y}}

\newcommand{\bz}{\mathbf{z}}

\DeclareSymbolFont{bsfletters}{OT1}{cmss}{bx}{n}
\DeclareSymbolFont{ssfletters}{OT1}{cmss}{m}{n}
\DeclareMathSymbol{\bsfGamma}{0}{bsfletters}{'000}
\DeclareMathSymbol{\ssfGamma}{0}{ssfletters}{'000}
\DeclareMathSymbol{\bsfDelta}{0}{bsfletters}{'001}
\DeclareMathSymbol{\ssfDelta}{0}{ssfletters}{'001}
\DeclareMathSymbol{\bsfTheta}{0}{bsfletters}{'002}
\DeclareMathSymbol{\ssfTheta}{0}{ssfletters}{'002}
\DeclareMathSymbol{\bsfLambda}{0}{bsfletters}{'003}
\DeclareMathSymbol{\ssfLambda}{0}{ssfletters}{'003}
\DeclareMathSymbol{\bsfXi}{0}{bsfletters}{'004}
\DeclareMathSymbol{\ssfXi}{0}{ssfletters}{'004}
\DeclareMathSymbol{\bsfPi}{0}{bsfletters}{'005}
\DeclareMathSymbol{\ssfPi}{0}{ssfletters}{'005}
\DeclareMathSymbol{\bsfSigma}{0}{bsfletters}{'006}
\DeclareMathSymbol{\ssfSigma}{0}{ssfletters}{'006}
\DeclareMathSymbol{\bsfUpsilon}{0}{bsfletters}{'007}
\DeclareMathSymbol{\ssfUpsilon}{0}{ssfletters}{'007}
\DeclareMathSymbol{\bsfPhi}{0}{bsfletters}{'010}
\DeclareMathSymbol{\ssfPhi}{0}{ssfletters}{'010}
\DeclareMathSymbol{\bsfPsi}{0}{bsfletters}{'011}
\DeclareMathSymbol{\ssfPsi}{0}{ssfletters}{'011}
\DeclareMathSymbol{\bsfOmega}{0}{bsfletters}{'012}
\DeclareMathSymbol{\ssfOmega}{0}{ssfletters}{'012}

\newcommand{\bdelta}{\bm{\delta}}

\newcommand{\bGamma}{\bm{\Gamma}}

\newcommand{\bSigma	}{\bm{\Sigma}}
\newcommand{\bPsi}{\bm{\Psi}}

\newcommand{\bXi}{\bm{\Xi}}

\newcommand{\bOmega}{\bm{\Omega}}
\newcommand{\bPhi}{\bm{\Phi}}
\newcommand{\bPi}{\bm{\Pi}}
\newcommand{\bTheta}{\bm{\Theta}}

\DeclareMathOperator{\Tr}{Tr}

\DeclareMathOperator{\eig}{eig}

\newcommand{\qednew}{\nobreak \ifvmode \relax \else
      \ifdim\lastskip<1.5em \hskip-\lastskip
      \hskip1.5em plus0em minus0.5em \fi \nobreak
      \vrule height0.75em width0.5em depth0.25em\fi}

\newcommand{\nn}{\nonumber\\}

\newcommand{\T}{^{\intercal}}

\newcommand{\tc}[1]{^{(#1)}}

\newcommand{\trace}[1]{{\Tr\left( #1 \right)}}

\newcommand{\cond}[2]{\left. {#1}\, \middle| \, {#2} \right.}

\DeclareDocumentCommand \P { g d() g } {%
	\IfNoValueTF {#3} 
	{%
		\IfNoValueTF {#1} 
		{%
			\IfNoValueTF {#2}
			{%
				\mathbb{P}%
			}%
			{%
				\mathbb{P}\left({#2}\right)%
			}%
		}%
		{%
			\IfNoValueTF {#2}
			{%
				\mathbb{P}_{#1}%
			}%
			{%
				\mathbb{P}_{#1}\left({#2}\right)%
			}%
		}%
	}%
	{%
		\IfNoValueTF {#1} 
		{%
			\mathbb{P}\left(\cond{#2}{#3}\right)%
		}%
		{%
			\mathbb{P}_{#1}\left(\cond{#2}{#3}\right)%
		}%
	}%
}

\DeclareDocumentCommand \E { g o g } {%
	\IfNoValueTF {#3} 
	{%
		\IfNoValueTF {#1} 
		{%
			\IfNoValueTF {#2}
			{%
				\mathbb{E}%
			}%
			{%
				\mathbb{E}\left[{#2}\right]%
			}%
		}%
		{%
			\IfNoValueTF {#2}
			{%
				\mathbb{E}_{#1}%
			}%
			{%
				\mathbb{E}_{#1}\left[{#2}\right]%
			}%
		}%
	}%
	{%
		\IfNoValueTF {#1} 
		{%
			\mathbb{E}\left[\cond{#2}{#3}\right]%
		}%
		{%
			\mathbb{E}_{#1}\left[\cond{#2}{#3}\right]%
		}%
	}%
}

\definecolor{gray90}{gray}{0.9}

\ifx\nohighlights\undefined

	\newcommand{\msout}[1]{\text{\color{green} \sout{\ensuremath{#1}}}}
	\newcommand{\del}[1]{{\color{green}\ifmmode \msout{#1}\else\sout{#1}\fi}}
\else

	\newcommand{\msout}[1]{#1}
	\newcommand{\del}[1]{#1}
\fi

\newcommand{\hide}[1]{}

\graphicspath{{./Figures/}} 
\ifx\diagnoselabel\undefined
	\relax
\else
	\makeatletter
	 \def\@testdef #1#2#3{%
		 \def\reserved@a{#3}\expandafter \ifx \csname #1@#2\endcsname
		\reserved@a  \else
	 \typeout{^^Jlabel #2 changed:^^J%
	 \meaning\reserved@a^^J%
	 \expandafter\meaning\csname #1@#2\endcsname^^J}%
	 \@tempswatrue \fi}
	\makeatother
\fi

\newcommand{\vq}{{\bf 1}_{|{\cal Q}|}}

\newcommand{\tbz}{\tilde{\bz}}
\newcommand{\tbP}{\tilde{\bP}}

\newcommand{\pC}{\left[{\bf C}_k\right]}
\newcommand{\pG}{\left[{\bf G}_{k+n|k}\right]}

\newcommand{\pT}{\left[\bT_k\right]}

\newcommand{\s}{{\setminus s}}

\newcommand{\JsIs}{{\calJ_s,\calI_s}}
\newcommand{\JlIl}{{\calJ_1,\calI_1}}
\newcommand{\JSIS}{{\calJ_S,\calI_S}}

\newcommand{\JnsIns}{{\calJ_\s,\calI_\s}}
\newcommand{\JnsJns}{{\calJ_\s,\calJ_\s}}

\newcommand{\IlIS}{{\calI_1,\calI_S}}
\newcommand{\ISIl}{{\calI_S,\calI_1}}
\newcommand{\IlIl}{{\calI_1,\calI_1}}
\newcommand{\IsIs}{{\calI_s,\calI_s}}
\newcommand{\IsJs}{{\calI_s,\calJ_s}}
\newcommand{\IsJns}{{\calI_s,\calJ_\s}}
\newcommand{\IsIns}{{\calI_s,\calI_\s}}
\newcommand{\InsIns}{{\calI_\s,\calI_\s}}

\newcommand{\submat}[2]{\left[#1\right]_{#2}}

\floatname{algorithm}{Procedure}

\begin{document}

\title{Compressive Privacy for a \\Linear Dynamical System}
\author{Yang~Song, Chong~Xiao~Wang, and Wee~Peng~Tay,~\IEEEmembership{Senior~Member,~IEEE}%
\thanks{
This work was supported by the ST Engineering NTU Corporate Lab through the NRF corporate lab@university scheme Project Reference C-RP10B, and the Singapore Ministry of Education Academic Research Fund Tier 2 grant MOE2018-T2-2-019. A preliminary version of this paper was presented at the IEEE Int. Conf. Acoustics, Speech, and Signal Processing, Calgary, Canada, 2018.  The authors are with the School of Electrical and Electronic Engineering, Nanyang Technological University, Singapore. E-mails: \texttt{{songy, wangcx, wptay}@ntu.edu.sg}
}%
}



\maketitle \thispagestyle{empty}


\begin{abstract}
We consider a linear dynamical system in which the state vector consists of both public and private states. One or more sensors make measurements of the state vector and sends information to a fusion center, which performs the final state estimation. To achieve an optimal tradeoff between the utility of estimating the public states and protection of the private states, the measurements at each time step are linearly compressed into a lower dimensional space. Under the centralized setting where all measurements are collected by a single sensor, we propose an optimization problem and an algorithm to find the best compression matrix. Under the decentralized setting where measurements are made separately at multiple sensors, each sensor optimizes its own local compression matrix. We propose methods to separate the overall optimization problem into multiple sub-problems that can be solved locally at each sensor. We consider the cases where there is no message exchange between the sensors; and where each sensor takes turns to transmit messages to the other sensors. Simulations and empirical experiments demonstrate the efficiency of our proposed approach in allowing the fusion center to estimate the public states with good accuracy while preventing it from estimating the private states accurately.    
\end{abstract}


\section{Introduction}
\label{sect:intro}

In the emerging Internet of Things (IoT) paradigm, large numbers of sensors are deployed in modern infrastructures, such as smart grids, population health monitoring, traffic monitoring, online recommendation systems, etc.\ \cite{SunSGC2017,ChenMBS2006,TaySPT0608,AlemdarCN2010,YuIET0714,VeugenICASSP15}. These sensors make observations and send data like the level of power consumption, sickness symptoms, or GPS coordinates in real-time to a fusion center \cite{TayITT0908,ChenSPT0312,LengJSP0115,TayJSTSP0315,HoSPT1015}, thus allowing utility companies and other service providers to improve their service offerings. However, some privacy-sensitive information such as personal activities, behaviors, preferences, habits and health conditions can be inferred from the raw measurements collected from individuals. For example, in a smart grid infrastructure, customers could be offered better rates if they continuously send their instantaneous power consumption to the utility company, which helps to improve the demand forecasting mechanism. However, from the fine-grained smart meter readings, the customers' private activities such as when and for how long appliances are used can be inferred \cite{HartProcIEEE1992,SankarSGT0613,HongIFST0917}. Another example is the inference of  individuals' private ratings or preferences from temporal changes in public recommendation systems \cite{NarISSP08,CalISSP11}. In traffic monitoring systems, individual users send anonymized personal location traces, which can be GPS coordinates measured by their smartphones, to a data aggregator to aid in traffic state estimation. However, an adversary may link an anonymous GPS trace to a particular person provided the knowledge of the person's residence and/or working location \cite{HohMCT0512,ZhangICMCN11,GisITST0615,AndreICASSP2017}. Therefore, it is essential to develop privacy preserving mechanisms to protect private state information from being inferred while retaining the quality of services that depend on the sensor measurements. All the above mentioned applications involve dynamic and time-varying data streams instead of static data sets.


In this paper, we consider a linear dynamical system (LDS) in which one or more sensors make measurements of the evolving state vector, and send information to a fusion center that makes the final state estimation. Suppose some of the states contain sensitive information. We call these states the private states, while the remaining states are public states. Our goal is to allow the fusion center to estimate the public states while making it difficult for it to estimate the private states. To achieve this, each sensor sanitizes its measurement to remove statistical information about the private states before sending it to the fusion center. A naive approach is for each sensor to first estimate the public states and send only this information to the fusion center. Such an approach does not prevent statistical inference of the private states at a future time step since in a LDS, the public states at one time step may provide statistical information about the private states in a future time step. Furthermore, when there are more than one sensor, it may not be possible for each individual sensor to estimate the public states. There is therefore a need to investigate the optimal way to sanitize information to achieve the best tradeoff between the utility of estimating the public states and the protection of the private states over all time steps. In this paper, we consider the use of a linear transformation or compression to map the sensor measurements into a lower dimensional space as the sanitization procedure. We call this \emph{compressive privacy} \cite{KunSPM2017}. 

\subsection{Related Work}\label{sect:related}

In general, privacy can be categorized into two classes: \emph{data} privacy and \emph{inference} privacy \cite{SunICASSP2016,HeICASSP2017,SunSPAWC2017,MengSun2018}. Data privacy protects the original measurements from being inferred by fusion center. The privacy metrics that have been proposed for preserving data privacy in a sensor network include homomorphic encryption \cite{BonICTC2005,IshICTC2007,CenASTC2009}, $k$-anonymity \cite{MacGehKif:C06}, plausible deniability \cite{BinShoGun:J17} and local differential privacy \cite{WanSIGKDD2003,SarACCCC2014,XioICASSP2016,Liao2017m,DucJorWai:C13}. Inference privacy on the other hand, prevents the fusion center from making certain statistical inferences. The privacy metrics that have been proposed to achieve inference privacy include information privacy \cite{CalAllerton2012,AsoISIT2016,HeICASSP2017,MengSun2018}, differential privacy \cite{dwork2014,CufYu:C16,GhoKle:arXiv16}, Blowfish privacy \cite{HeMacDin:C14}, mutual information privacy \cite{WangTIT2016,duchi2014} and average information leakage \cite{SalGCSIP2013,YamTIT1183}. The relationship between data privacy and inference privacy has been studied in \cite{WangTIT2016,SunSPAWC2017}. 

The type of privacy we consider in this paper can be considered to be a form of inference privacy since we wish to prevent the fusion center from inferring about the private states. The aforementioned inference privacy metrics like information privacy, differential privacy, mutual information privacy and average information leakage either assume a finite state alphabet (and are thus more applicable in a \emph{hypothesis testing} context instead of an estimation context), or do not directly guarantee that the fusion center cannot estimate the private states to within a certain error, making the choice of the privacy budgets in these metrics unclear and unintuitive. Furthermore, although quantities like mutual information have relationships with estimation error, they are not as easy to work with in a LDS where we consider privacy over multiple time steps. Therefore, in this paper, we define privacy on the estimation error variance instead, and require that the estimation error variances of the private states to be above a predefined threshold.

The papers \cite{GloACGKB2003,CheANIPS2002,CheNIPS2003,EmadITW2013} proposed to extract the relevant aspects of the data by the information bottleneck (IB) approach. Given the joint distribution of a ``source" variable $\bz$ and another ``relevance" variable ${\bf y}$, IB compresses $\bz$ to obtain $\tbz=\bC\bz$, while preserving information about ${\bf y}$. IB tradeoffs the complexity of the representation of $\bz$, measured by the mutual information $I(\bz; \tbz)$ , against the accuracy of this representation measured by $I(\by; \tbz)$. The compression is designed to make $I(\bz; \tbz)$ small and $I(\by;\tbz)$ large. We can interpret $\by$ to be the public state or information while $\bz$ to be the private state.  The privacy funnel (PF) method \cite{MakITW2014} uses a mapping from the source variable $\bz$ to $\tbz=\bC\bz$ so that $I(\by;\tbz)$ is small and $I(\bz; \tbz)$ is large. PF operates in a way that is opposite to IB by treating $\by$ as the private state. Both IB and PF do not consider information over multiple time steps, and are not directly applicable to a LDS without additional prior information about the LDS evolution.

The papers \cite{Diamantaras2016,KunSPM2017,KunJFI2017,Al2017,ChaChaMit:C16, KunChaCha:J17,ChaChaKun:C17} introduced the concept of compressive privacy (CP), which is a dimension-reducing subspace approach. They considered the case where the source variable $\bz$ can be mapped as $\bGamma^{(\calP)}\bz$ to the utility subspace and as $\bGamma^{(\calQ)}\bz$ to the privacy subspace, where $\bGamma^{(\calP)}$ and $\bGamma^{(\calQ)}$ are projection matrices. Under the assumption that $\bz$ has a Gaussian distribution, the compression matrix $\bC$ is designed to achieve an optimal privacy-utility tradeoff based on $I(\bGamma^{(\calP)}\bz; \tbz)$ and $I(\bGamma^{(\calQ)}\bz; \tbz)$, where $\tbz=\bC\bz$. In the case where the projection matrices $\bGamma^{(\calP)}$ and $\bGamma^{(\calQ)}$ are unknown, a machine learning approach is proposed to learn the compression matrix $\bC$ from a set of training data. Our formulation is similar to \cite{KunSPM2017}, except that we assume that the underlying state vector generating the measurement $\bz$ can be divided into public and private states, instead of the measurement $\bz$ being mapped into utility and privacy subspaces. A more detailed technical discussion of the differences is provided in \cref{section:pf}. Furthermore, different from \cite{Diamantaras2016,KunSPM2017,KunJFI2017,Al2017} which considers $\bz$ to be a single-shot observation, we consider a LDS in which observations are temporally correlated, and our goal is to preserve the privacy of the private states over multiple future time steps.  

The references \cite{NyTAC2014,DegueGlobalSIP17} developed differential privacy mechanisms for the measurements $\bz_k$, $k=1,2,\ldots$, in a LDS. The authors of \cite{NyTAC2014} proposed to use an input perturbation mechanism so that $\tbz_k=\bz_k+\bm{\nu}_k$, where $\bm{\nu}_k$ is white Gaussian noise, is used in place of $\bz_k$ at each time step $k$ to guarantee $(\epsilon,\delta)$-differential privacy for the original measurement $\bz_k$. The paper \cite{DegueGlobalSIP17} proposed to apply an input perturbation mechanism together with a linear transformation, i.e.,  $\tbz_k=\bC\bz_k+\bm{\nu}_k$. By adding white Gaussian noise $\bm{\nu}_k$ according to the Gaussian mechanism \cite{DworkEUROCRYPT2006}, $\bC\bz_k$ is differentially private. The transformation matrix  $\bC$ is designed to minimize the mean-square error (MSE) of the states. These papers' objectives are to preserve the differential privacy of the system measurements, which is different from our goal of preventing the statistical inference of a subset of states.

The authors in \cite{SteSPL2014} proposed to project an observation $\bz$ into a lower dimensional space to obtain $\tbz=\bC\bz$. The transformation $\bC$ is designed to maximally retain the estimation accuracy of the entire system's state and the dimension of the space projected into is predefined. This work also does not consider preserving the estimation privacy of a subset of states in a LDS.

\subsection{Our Contributions}

In this paper, we consider the use of a compressive linear transformation on the measurements at one or multiple sensors in a LDS to prevent the fusion center from estimating a set of private states with low error while still allowing it to estimate a set of public states with good accuracy. Our main contributions are as follows. 
\begin{enumerate}[(i)]
\item We formulate a utility-privacy tradeoff optimization problem for a LDS involving privacy constraints on the predicted estimation error of the private states multiple steps ahead. We show how to find the dimension of the compressive map and the map itself and propose an algorithm to achieve this. We provide a bound for the number of steps to look ahead to achieve the same privacy level in all future time steps.
\item We consider the case where there are multiple decentralized sensors that optimize their own local compressive map. We propose optimization methods for the cases where 1) there are no message exchanges between the sensors; and 2) each sensor takes turns to transmit messages to the other sensors.
\item We present extensive simulation results that demonstrate that imposing privacy constraints multiple steps ahead is necessary in some LDSs, and examine the impact of different choices of the state evolution and measurement matrices on the utility-privacy tradeoff.  We also verify the performance of our proposed approach on an empirical ultra wideband (UWB) localization system and human activity recognition data set.
\end{enumerate}

\begin{table*}[!t]
\centering
\caption{Summary of commonly-used symbols.}\label{tab:notations}
\begin{tabular}{|c|l|}
\hline
Notation & Definition \\
\hline
$\calP=\{1,\ldots,|\calP|\}$ & index set of the $|\calP|$ public states \\
\hline
$\calQ=\{|\calP|+1,\ldots,L\}$ & index set of the $|\calQ| = L - |\calP|$ private states \\
\hline
$\bx_k$, $\bz_k$ & state and measurement vector at time $k$, respectively\\
\hline
$\bF_k\in\Real^{L\times L}$, $\bH_k\in\Real^{N\times L}$ & state evolution and measurement matrices at time $k$, \cref{system} \\
\hline
$\bQ_k$, $\bR_k$ & state and measurement noise covariances at time $k$, \cref{system}\\
\hline
$\bC_k\in\Real^{M\times N}$ & compression matrix applied to the measurement at time $k$, \cref{eq:obs_compressed} \\
\hline
$\tbP_{k+n|k}$ & $n$-step prediction error covariance matrix {\it after compression} at time $k$, \cref{eq:p_k_plus_n_given_k} \\
\hline
$\bT_k$ &  $={\bf H}_k \tbP_{k|k-1}{\bf H}\T_k+{\bf R}_k$, \cref{bT}\\
\hline
${\bf G}_{k+n|k}$ & $={\bf H}_k \tbP_{k|k-1}\bF_{k+1:k+n}\T$, \cref{bG} \\
\hline
$\bD_{k+n|k}$ & $\tbP_{k+n|k}-\tbP_{k+n|k-1}$, error reduction due to measurement made at time $k$, \cref{D} \\
\hline
$\tau_k(\bC_k)$ & public error trace at time $k$ , \cref{eq:utility_func}\\
\hline
$\eta_{k+n|k}(\bC_k)$ & $n$-step look-ahead private error function at time $k$ with $\eta_k=\eta_{k|k}$, \cref{eq:privacy_multi} \\
\hline
$u_k(\bC_k)$ & utility at time $k$, \cref{eq:utility_gain} \\
\hline
$\ell_{k+n|k}(\bC_k)$ & $n$-step look ahead privacy loss at time $k$, \cref{eq:privacy_loss_kpn} \\
\hline
\end{tabular}
\end{table*}

A preliminary version of this work was presented in \cite{SongICASSP2018} in which only the single sensor case with a single step look-ahead privacy constraint was considered. This paper generalizes \cite{SongICASSP2018} to the case where there are multiple decentralized sensors with multi-step look-ahead privacy constraints. New theoretical insights and numerical results are also presented in this paper.

The rest of this paper is organized as follows. In \cref{section:pf}, we present our problem formulation, assumptions and an optimization framework to achieve an optimal utility-privacy tradeoff. In \cref{section:central}, we propose a centralized solution to find the compression matrix that optimizes the utility-privacy tradeoff at each time step. In \cref{section:dist}, we consider the decentralized case where multiple sensors are involved and proposed different optimization approaches. We present simulation results in \cref{section:simulations}, and conclude in \cref{section:c}.

{\em Notations}: We use $\Real$ to denote the set of real numbers. The normal distribution with mean $\mu$ and variance $\sigma^2$ is denoted as ${\cal N}(0,\sigma^2)$. $\min(\ba)$ is the minimum element in the vector $\ba$, $\bm{1}_n$ is the vector of all 1's of length $n$, and $\bI_n$ is the identity matrix of size $n\times n$. We use $\T$ to represent matrix transpose, ${\rm diag}({\bf A},{\bf B},\ldots)$ to represent a block diagonal matrix with submatrices ${\bf A},{\bf B},\ldots$ being the diagonal elements, ${\rm Tr}(\cdot)$ to denote the trace operation, and ${\rm vecdiag}(\bA)$ to denote a column vector consisting of diagonal entries of $\bA$. We use $\submat{\bA}{\calP,\calQ}$ to denote the sub-matrix of matrix ${\bf A}$ consisting of the entries ${\bf A}(i,j)$ for all $(i,j)\in\calP\times \calQ$, $\left[{\bf A}\right]_{:,\calQ}$ to denote the sub-matrix of matrix ${\bf A}$ consisting of columns indexed by $\calQ$, and $\left[{\bf A}\right]_{\calP}$ to denote the sub-matrix of matrix ${\bf A}$ consisting of rows indexed by $\calP$. Given a set of $n\times n$ matrices ${\bf A}_i,{\bf A}_{i+1},\ldots, {\bf A}_{j}$, we use $\bA_{i:j}$ to denote the matrix product $\bA_j \bA_{j-1} \cdots \bA_i$ if $i\leq j$, and define $\bA_{i:j}=\bI_n$ if $i>j$. The notation ${\bf A}\succeq{\bf 0}$ means that ${\bf A}$ is positive semidefinite. Let $\bA^\perp$ denote the basis matrix of the null space of $\bA$, i.e., $\bA\T \bA^\perp ={\bf 0}$. For easier reference, we summarize some of the commonly-used symbols in \cref{tab:notations}.

\section{Problem Formulation}
\label{section:pf}

We consider a LDS given by the following state and measurement equations at time step $k$:
\begin{subequations}\label{system}
\begin{align}
{\bf x}_k&={\bf F}_k {\bf x}_{k-1} + {\bf v}_{k}, \label{eq:dyn} \\
\bz_k&={\bf H}_k {\bf x}_{k} + {\bf n}_{k}, \label{eq:obs}
\end{align}
\end{subequations}
where 
${\bf x}_k\in \Real^{L}$ and $\bz_k\in \Real^{N}$ are the system's state and measurement at time $k$, respectively. The state and process noise ${\bf v}_{k}$ and ${\bf n}_{k}$ are independent, and follow zero-mean Gaussian distributions with positive definite covariances ${\bf Q}_{k}$ and ${\bf R}_{k}$, respectively. We assume that the state evolution matrices ${\bf F}_k\in \Real^{L\times L}$ are known for all $k$, and the measurement matrices $\{{\bf H}_j\in \Real^{N\times L} : j \leq k\}$, where $N\geq L$ are known only up to the current time step $k$. This assumption is made because in many applications like target tracking \cite{ReidTAC79,JulierJPROC04,HuaDis:J07}, the measurement matrix $\bH_k$ is chosen adaptively at each time step $k$. 

Another example is estimation over lossy networks where the measurement matrices are time-varying \cite{SinSchFra:J04, SchSinFra:J07}. On their routes to the gateway, sensor packets, possibly aggregated with measurements from several nodes, may become intermittent because of time-varying transmission intervals or delays \cite{DonHeeWou:J11, XiaShaChe:J09}, packet dropouts \cite{LiuGol:C04, HuaDey:J07, LiaChePan:J10, MoSin:J12, NouLeoDey:J14, SahCheSha:J07, ShiEpsMur:J10, SilSol:J13, XiaShaChe:J09}, random message exchanges depending on the availability of appropriate network links \cite{KarMou:J11}, fading channels \cite{QueAhlJoh:J13}, and other communication constraints \cite{RohMarFu:J14}. The measurement matrix $\bH_k$ is thus unknown until the sensor measurements are received at time $k$.


To obtain the minimum mean square estimate of the system state in \cref{system}, it is well known \cite{Kalman1960} that the Kalman filter is optimal. The Kalman filter contains two distinct phases: ``predict" and ``update". In the ``predict" phase, the state estimate and error covariance are predicted, respectively, by
\begin{align}
\hat{\bf x}_{k|k-1} &= {\bf F}_k \hat{\bf x}_{k-1|k-1}, \label{eq:pred_state}\\
\bP_{k|k-1} &= {\bf F}_k \bP_{k-1|k-1} {\bf F}\T_k + {\bf Q}_k.\label{eq:pred_err_cov}
\end{align}
In the ``update" phase, the state estimate and error covariance are updated, respectively, through
\begin{align}\label{eq:state}
\hat{\bf x}_{k|k} &= \hat{\bf x}_{k-1|k-1} + {\bf K}_k (\bz_k-{\bf H}_k\hat{\bf x}_{k|k-1}), \\
\bP_{k|k} &= ({\bf I}-{\bf K}_k {\bf H}_k) \bP_{k|k-1},\label{eq:err_cov}
\end{align}
where ${\bf K}_k=\bP_{k|k-1} {\bf H}_k\T \left({\bf H}_k\bP_{k|k-1}{\bf H}\T_k+{\bf R}_k\right)^{-1}$ denotes the Kalman gain.

In this paper, we consider the case where the state ${\bf x}_k$ may be partitioned into two parts as \begin{align}
\bx_k
&=\begin{bmatrix}
\submat{\bx_k}{\calP}\\
\submat{\bx_k}{\calQ}
\end{bmatrix},
\end{align}
where $\left[{\bf x}_k\right]_\calP\in \Real^{|\calP|}$, with $\calP=\{1,\ldots,|\calP|\}$ and $|\calP| < L$, contains the public states that are to be estimated, while $\left[{\bf x}_k\right]_\calQ\in \Real^{|\calQ|}$, with $\calQ=\{|\calP|+1,\ldots,L\}$ and $|\calQ| = L-|\calP| > 0$, are the private states containing sensitive information that we wish to protect. Our goal is to minimize the estimation errors of the public states $\left[\hat{\bf x}_{k|k}\right]_\calP$, while ensuring that the estimation errors of private states $\left[\hat{\bf x}_{k|k}\right]_\calQ$ are above a predefined threshold.

In \cite{KunSPM2017}, it is assumed that the measurement or feature vector $\bz_k$ can be mapped as $\bGamma^{(\calP)}\bz_k$ to a utility subspace and as $\bGamma^{(\calQ)}\bz_k$ to a privacy subspace. Our formulation is equivalent to a variant of the formulation in \cite{KunSPM2017} if there is no additive noise $\bn_k$ in \cref{eq:obs}, and $[\bH_k]_{:,\calP}$ is orthogonal to $[\bH_k]_{:,\calQ}$, but is in general different from \cite{KunSPM2017}. The advantage of our formulation is that in applications with a known LDS model, it is easier to specify which system states are public and private directly instead of through the measurements or observed features. Furthermore, to apply the formulation from \cite{KunSPM2017} to a LDS where we want to protect the privacy of some states over multiple time steps, will require prior knowledge of the measurement matrices $\bH_j$ for $j \geq k$, where $k$ is the current time step. In particular, such an approach is impractical if $j$ is large.

To prevent the fusion center from inferring the private states $[\bx_k]_{\calQ}$, we assume that a linear mapping or compression matrix $\bC_k\in\Real^{M\times N}$, where $1\leq M \leq N$, is applied to the measurement to obtain 
\begin{align}\label{eq:obs_compressed}
\tbz_k=\bC_k\bz_k&=
\bC_k{\bf H}_k {\bf x}_{k} + \bC_k{\bf n}_{k}.
\end{align}
Let $\tbP_{k|k}$ be the state error covariance based on the measurements $\tbz_{k'}$, where $k'\leq k$. %

A smaller $\tbP_{k|k}(i,i)$ implies that the $i$-th state can be estimated with lower error on average. Therefore, as the utility, we aim to minimize the sum of the expected estimation errors of the public states or the public error trace at time $k$ defined as
\begin{align}\label{eq:utility_func}
\tau_k\left(\bC_k\right) = {\rm Tr}\left(\left[\tilde\bP_{k|k}\right]_{\cal P,P}\right).
\end{align}
On the other hand, the private error function at time $k$ is defined as
\begin{align}\label{eq:privacy_func}
\eta_k\left(\bC_k\right) =  {\cal F} \left({\rm vecdiag}\left(\left[\tilde\bP_{k|k}\right]_{\cal Q,Q}\right)\right),
\end{align}
where $\calF(\bx)=\bA\bx:\Real^{|\calQ|}\rightarrow \Real^{|\calF|}$ is a user-defined linear map such that $\calF(\ba) \leq \calF(\bb)$ if $\ba\leq\bb$, with $\leq$ here denoting element-wise inequality. The non-decreasing property follows iff $\bA$ has non-negative entries. 

For example, if ${\cal F}(\bx)={\bf 1}^T_{|{\cal Q}|} \bx$ with $|\calF|=1$, the private error function is the sum of the private states' error variances and
\begin{align}\label{eq:eqk1}
\eta_k\left(\bC_k\right)={\rm Tr}\left(\left[\tilde\bP_{k|k}\right]_{\cal Q,Q}\right).
\end{align} 
If the privacy of \emph{every} private state is important, we can choose ${\cal F}(\bx)=\bI_{\calQ}\bx$ and $|\calF|=|\calQ|$. Then, the private error function becomes
\begin{align}\label{eq:eqk2}
{\bm \eta}_k\left(\bC_k\right)={\rm vecdiag}\left(\left[\tilde\bP_{k|k}\right]_{\cal Q,Q}\right).
\end{align}
At each time $k$, we seek to optimize the privacy-utility tradeoff myopically as follows:
\begin{align}\label{eq:1step_problem}\tag{P0} 
\min_{\bC_k} &~ \tau_k(\bC_k), \nonumber\\
{\rm s.t.}&~\eta_k(\bC_k) \geq {\cal F}(\delta {\bf 1}_{|{\cal Q}|}), \nonumber
\end{align}
where $\delta$ is a predefined threshold, and %
the minimization is over the set of compression matrices $\{\bC_k : \bC_k\in\Real^M,\ 1\leq M \leq N\}$. Note that we are optimizing over $\bC_k$ as well as its dimension, i.e., $M$, at every time step and \cref{eq:1step_problem} is solved sequentially for each time step $k$.

The private error function given in \cref{eq:eqk2} is more restrictive than that given in \cref{eq:eqk1}. In many practical problems, the private information depends on all the private states so that protecting  some (but not necessarily all) of the private states may be sufficient to protect the overall private information. For example, to protect someone's geo-location information, which consists of $x$-, $y$- and $z$- coordinates, it may be sufficient in some applications to obfuscate just one or two coordinates instead of all three to achieve reasonable geo-location privacy. This motivates the use of the private error trace in \cref{eq:eqk1} as one potential privacy measure.

Problem \cref{eq:1step_problem} focuses on the privacy-utility tradeoff at the current time step without taking into account the future time steps.  As the predicted error covariance $\tilde\bP_{k+1|k}$ relates to the error covariance $\tilde\bP_{k|k}$ through \cref{eq:pred_err_cov}, i.e., 
\begin{align}\label{tbP_k_given}
\tilde\bP_{k+1|k} = {\bf F}_{k+1} \tilde\bP_{k|k} {\bf F}\T_{k+1} + {\bf Q}_{k+1},
\end{align}
$\bC_k$ affects not just the privacy-utility tradeoff at the current time step $k$, but also the tradeoff at future time steps. This is illustrated in an example below.

\begin{Example}\label{ex:flip}
Consider the case where $\left[{\bf F}_k\right]_{\cal P,P}=0$ and $\left[{\bf F}_k\right]_{\cal Q,Q}=0$, i.e.,
\begin{align*}
{\bf F}_k&=
\left[
\begin{array}{cc}
{\bf 0} & \left[{\bf F}_k\right]_{\cal P,Q}\\
\left[{\bf F}_k\right]_{\cal Q,P} &  {\bf 0}
\end{array}
\right].
\end{align*}
Recall that the prediction error covariance $\bP_{k|k-1}$ is related to $\bP_{k-1|k-1}$ via \cref{eq:pred_err_cov}. Substituting ${\bf F}_k$ into  \cref{eq:pred_state} and \cref{eq:pred_err_cov} yields
\begin{align*}
\left[
\begin{array}{c}
\left[\hat{\bf x}_{k|k-1}\right]_\calP \\
\left[\hat{\bf x}_{k|k-1}\right]_\calQ
\end{array}\right]
 &= 
\left[
\begin{array}{c}
\left[{\bf F}_k\right]_{\cal P,Q}\left[\hat{\bf x}_{k-1|k-1}\right]_\calQ \\
\left[{\bf F}_k\right]_{\cal Q,P}\left[\hat{\bf x}_{k-1|k-1}\right]_\calP
\end{array}\right], \\
\bP_{k|k-1} &= \left[
\begin{array}{cc}
\left[{\bf F}_k\right]_{\cal P,Q}  \left[\bP_{k-1|k-1}\right]_{\cal Q,Q} \left[{\bf F}_k\right]_{\cal Q,P} & *\\
* &  \left[{\bf F}_k\right]_{\cal Q,P}  \left[\bP_{k-1|k-1}\right]_{\cal P,P} \left[{\bf F}_k\right]_{\cal P,Q}
\end{array}
\right] + {\bf Q}_k.
\end{align*}
We see that the state evolution matrix ${\bf F}_{k}$ converts the public state at time $k-1$ into a linear function of the private state at time $k$, and vice versa. This means that if a low public error trace is achieved at time $k-1$, then a low private error function value at time $k$ is inevitable if \cref{eq:1step_problem} is used to design the utility-privacy tradeoff. This example shows that it is necessary to incorporate privacy constraints for future time steps into \cref{eq:1step_problem}. 
\end{Example}

Since the dynamical model \cref{eq:dyn} is publicly known, we may predict the system's state $n$ time steps in the future. The $n$-step prediction error covariance $\tbP_{k+n|k}$ at time $k$ is given by 
\begin{align}\label{eq:p_k_plus_n_given_k}
\tilde\bP_{k+n|k} 
&= \bF_{k+1:k+n}\tilde\bP_{k|k}{\bF\T_{k+1:k+n}} + \sum_{i=1}^{n} {\bf F}_{k+i+1:k+n}{\bf Q}_{k+i}{{\bf F}\T_{k+i+1:k+n}}.
\end{align}
The \emph{look-ahead} private error function at time $k+n$ can be defined as
\begin{align}\label{eq:privacy_multi}
\eta_{k+n|k}(\bC_k) = {\cal F}\left({\rm vecdiag}\left(\left[\tilde\bP_{k+n|k}\right]_{\cal Q,Q}\right)\right).
\end{align}

Incorporating privacy constraints on the prediction error covariance $r_k$ steps ahead, we have the following optimization problem at each time step $k$:
\begin{align}\label{eq:multi-step_problem}\tag{P1}
\min_{\bC_k} &\  \tau_k(\bC_k),\\
{\rm s.t.} &\ \eta_{k+n|k}(\bC_k)\geq {\cal F}(\delta {\bf 1}_{|{\cal Q}|}),\ n=0,\ldots,r_k, \nonumber
\end{align}%
where $\eta_{k|k}=\eta_k$ in \cref{eq:privacy_func}. When $r_k=0$, \cref{eq:multi-step_problem} becomes \cref{eq:1step_problem}. Note again that \cref{eq:multi-step_problem} is solved sequentially or myopically at each time step $k$, with the minimization over $\{\bC_k : \bC_k\in\Real^M,\ 1\leq M \leq N\}$.

Let 
\begin{align}
\bT_k&=\bH_k \tilde{\bP}_{k|k-1}\bH\T_k+\bR_k, \label{bT}
\end{align}
and for $n\geq 0$, let 
\begin{align}
\bG_{k+n|k}&=\bH_k \tbP_{k|k-1}\bF\T_{k+1:k+n}.\label{bG}
\end{align}
We define 
\begin{align}
\bD_{k+n|k}
&=\bF_{k+1:k+n}\bD_{k|k}\bF\T_{k+1:k+n},\label{Dn}
\end{align}
where
\begin{align}
\bD_{k|k} &= \tbP_{k|k-1}-\tbP_{k|k} \nn
&= \bG_{k|k}\T\bC\T_k \left(\bC_k\bT_k\bC\T_k\right)^{-1} \bC_k \bG_{k|k} \label{D0}
\end{align}
is the reduction of the error covariance due to the measurement made at time $k$. From \cref{Dn,D0}, we obtain
\begin{align}
\bD_{k+n|k}&= \bG_{k+n|k}\T \bC_k\T \left(\bC_k\bT_k\bC_k\T\right)^{-1} \bC_k \bG_{k+n|k}.\label{D}
\end{align}
From \cref{eq:p_k_plus_n_given_k} and \cref{Dn}, we have 
\begin{align}\label{tbP_kn_given_k_D}
\tbP_{k+n|k} = \tbP_{k+n|k-1} - \bD_{k+n|k}.
\end{align}
Replacing $n$ by $n+1$ and $k$ by $k-1$ in \cref{eq:p_k_plus_n_given_k}, and applying \cref{tbP_k_given}, we obtain
\begin{align}\label{P_kpn_km1}
\tbP_{k+n|k-1}
=&{\bf F}_{k+1:k+n}\tbP_{k|k-1}\bF\T_{k+1:k+n} + \sum_{i=1}^{n} {\bf F}_{k+i+1:k+n}\bQ_{k+i}\bF\T_{k+i+1:k+n}.
\end{align}
Therefore, $\tbP_{k+n|k-1}$ does not depend on $\bC_k$. We can now define the {\em utility} to be
\begin{align}\label{eq:utility_gain}
u_k(\bC_k) &= \trace{\submat{\tbP_{k|k-1}}{\calP,\calP}} - \tau_k(\bC_k) = {\rm Tr}\left(\left[\bD_{k|k}\right]_{\cal P,P} \right), 
\end{align}
and the $n$-step look-ahead \emph{privacy loss} function at time $k$ as  
\begin{align}
\ell_{k+n|k}(\bC_k) 
=&\ {\cal F}\left( {\rm vecdiag}\left(\left[\tbP_{k+n|k-1}\right]_{\cal Q,Q} \right)\right) - \eta_{k+n|k}(\bC_k)\nn
=&\ {\cal F}\left( {\rm vecdiag}\left(\left[\bD_{k+n|k}\right]_{\cal Q,Q} \right)\right).\label{eq:privacy_loss_kpn}
\end{align}
Problem \cref{eq:multi-step_problem} can then be equivalently recast as
\begin{align}\label{eq:multi-step_problem_v2}\tag{P2} 
\max_{\bC_k} &~ u_k(\bC_k) \nonumber\\
{\rm s.t. }&~ \ell_{k+n|k}(\bC_k) \leq {\cal F} \left(\bar{\bdelta}_{k+n} \right),\ n=0,\ldots,r_k, \nonumber
\end{align}
where 
\begin{align}\label{bardelta}
\bar{\bdelta}_{k+n}={\rm vecdiag}\left(\left[\tilde\bP_{k+n|k-1}\right]_{\cal Q,Q} \right) - \delta {\bf 1}_{|{\cal Q}|}
\end{align}
is the $n$-step \emph{privacy loss threshold}. From \cref{ex:flip}, we see that a sufficiently large $r_k$ at the current time step $k$ is required to ensure that there exists a $\bC_{t}$ such that $\eta_{t|t}(\bC_{t})\geq {\cal F}(\delta {\bf 1}_{|{\cal Q}|})$ for all future time steps $t > k$. In the following, we provide a lower bound for $r_k$ under different assumptions. We start off with an elementary lemma.

\begin{Lemma}\label{lemma:ABA}
Consider a matrix $\bA$ and square matrix $\bB$. Suppose $\bA\bA\T \succeq \alpha\bI$ and $\bB\succeq \beta\bI$, then $\bA \bB \bA\T \succeq \alpha \beta \bI $, where $\alpha$ and $\beta$ are non-negative scalars.
\end{Lemma}
\begin{IEEEproof}
Since $\bB-\beta\bI \succeq {\bf 0}$, we have $\bA \left(\bB-\beta\bI\right) \bA\T \succeq {\bf 0}$. From $\bA \bA\T - \alpha \bI \succeq {\bf 0}$, we obtain $\bA \bB \bA\T -\alpha \beta \bI \succeq {\bf 0}$, which completes the proof.
\end{IEEEproof}

\begin{Proposition}\label{prop:min_n}
Suppose that at the current time step $k$, $\tbP_{k|k-1}\succeq \nu_k\bI$ for some $\nu_k>0$. Suppose also that for all $t \geq k$, $\bQ_t\succeq \epsilon\bI$ for some $\epsilon>0$, and $\bF_t\bF_t\T \succeq \xi\bI$ for some $\xi>0$. Then if either 
\begin{enumerate}[(i)]
\item $\xi\geq 1$, or
\item $\xi<1$, ${\cal F} (\delta \vq)  \leq \frac{\epsilon}{1-\xi} {\cal F}\left(\vq\right) $, and $\nu_k < \frac{\epsilon}{1-\xi}$, 
\end{enumerate}
and \footnote{We define $\left.\frac{1-\xi^r}{1-\xi}\right|_{\xi=1}=\underset{\xi\to1}{\lim}\frac{1-\xi^r}{1-\xi}=r$.} 
\begin{align}
r_k \geq&\ r_{k,\star} \nonumber \\
\triangleq&\ \underset{r\geq0}{\min} \left\{r:\ {\cal F}\left(\vq\right) \left(\epsilon\frac{1-\xi^r}{1-\xi} + \xi^r\nu_k\right) \geq {\cal F}(\delta \vq)\right\}  -1,\label{rkstar}
\end{align}
there exists $\bC_t$ such that $\eta_{t'|t}(\bC_t)\geq {\cal F} (\delta \vq)$, for all $t'\geq t \geq k$.
\end{Proposition}
\begin{IEEEproof}

From \cref{tbP_kn_given_k_D}, we have for all $t'\geq t \geq k$, $\tbP_{t'|k-1} \succeq \tbP_{t'|t}$.
If 
\begin{align}\label{tbP_t_k_condition}
\submat{\tbP_{t'|k-1}}{\calQ,\calQ} \succeq \delta \bI_{|\calQ|},
\end{align}
we can always choose $\bC_t={\bf 0}$ so that $\submat{\tbP_{t'|t}}{\calQ,\calQ} = \submat{\tbP_{t'|k-1}}{\calQ,\calQ} \succeq \delta \bI_{|\calQ|}$, which along with the non-decreasing property of $\calF$ implies that $\eta_{t'|t}(\bC_t)\geq {\cal F}(\delta \vq)$. We next show that the conditions given in the proposition statement lead to \cref{tbP_t_k_condition}.

From \cref{eq:p_k_plus_n_given_k}, after some minor manipulations using \cref{tbP_k_given}, we can rewrite $\tbP_{t'|k-1}$ as
\begin{align}\label{Ptk}
\tbP_{t'|k-1} =&\ {\bf F}_{k+1:t'}\tbP_{k|k-1}\bF\T_{k+1:t'} + \sum_{i=1}^{n} {\bf F}_{k+i+1:t'}\bQ_{k+i}\bF\T_{k+i+1:t'}.
\end{align}
Applying \cref{lemma:ABA} to each term on the right hand side of \cref{Ptk}, we obtain
\begin{align*}
\tbP_{t'|k-1}\succeq \left(\xi^r \nu_k + \sum_{i=1}^r \xi^{r-i}\epsilon \right) \bI,
\end{align*}
where $r = t'-k$. We therefore have
\begin{align*}
{\cal F}\left({\rm vecdiag}\left({\submat{\tbP_{t'|k-1}}{\calQ,\calQ}}\right)\right) 
\geq \ f(r) \triangleq {\cal F}\left(\vq \right)
\left(\epsilon \frac{1-\xi^r}{1-\xi} + \xi^r\nu_k \right).
\end{align*}
If $\xi\geq1$, $f(r)\uparrow\infty$ as $r\to\infty$. On the other hand, if $\xi<1$, we have 
\begin{align*}
f(r+1)-f(r)= {\cal F}\left(\vq\right)\left(\frac{\epsilon}{1-\xi}-\nu_k\right)(\xi^r-\xi^{r+1}),
\end{align*}
and $f(r)\uparrow {\cal F}\left(\vq\right) \frac{\epsilon}{1-\xi}$ as $r\to\infty$ if $\nu_k < \epsilon/(1-\xi)$. Therefore, under the conditions in the proposition, by choosing the number of look-ahead steps $r_k$ to satisfy $f(r_k+1) \geq {\cal F}(\delta \vq)$, $\eta_{t'|t}(\bC_t)\geq {\cal F}(\delta \vq)$ holds for all $t'\geq t \geq k$, and the proposition is proved. 
\end{IEEEproof}

\section{A Single Sensor}\label{section:central}

In this section, we consider the case where there is only a single sensor so that \cref{eq:multi-step_problem_v2} can be optimized over all $\bC_k \in \Real^{M\times N}$ and $1\leq M\leq N$ in a centralized setting. This gives the best utility-privacy tradeoff and serves as a benchmark for the decentralized case that we consider in \cref{section:dist}. We first prove an elementary lemma.

\begin{Lemma}\label{lemma:U=A}
Suppose $\bA\in\Real^{M\times N}$ with $M \leq N$ has full row rank, and $\bB\in\Real^{N\times N}$ is a positive definite matrix. Then,
\begin{align}\label{eq:ABA}
\bA\T\left(\bA \bB \bA\T \right)^{-1}\bA 
&= \bB^{-1/2}\bU\bU\T\bB^{-1/2}, 
\end{align}
where $\bU\in\Real^{N\times M}$ consists of the $M$ right unit singular vectors associated with the $M$ non-zero singular values of $\bA\bB^{1/2}$. 
\end{Lemma}
\begin{IEEEproof}
See Appendix~\ref{app:proof_lemma:U=A}.
\end{IEEEproof}

Suppose that $\bC_k$ is full row rank (otherwise we can choose a smaller $M$). From $\bD_{k+n|k}$ defined in \cref{D} and \cref{lemma:U=A}, we obtain
\begin{align}
\bD_{k+n|k}
&= \bG_{k+n|k}\T \bT_k^{-1/2} \bU_k \bU_k\T \bT_k^{-1/2} \bG_{k+n|k}, \label{eq:p_kpn_given_k}
\end{align}
where $\bU_k$ consists of the $M$ right unit singular vectors of $\bC_k\bT_k^{1/2}$ associated with its non-zero singular values. For any $n\geq 0$ and index set $\calI$, let 
\begin{align}\label{bTheta}
\bTheta^{(\calI)}_{k+n|k}&=\submat{\bT_k^{-1/2} \bG_{k+n|k}}{:,\calI} \submat{\bT_k^{-1/2} \bG_{k+n|k}}{:,\calI}\T.
\end{align}
Then, the utility function and privacy loss function in \cref{eq:utility_gain} and \cref{eq:privacy_loss_kpn} can be expressed as functions of $\bU_k$ as follows:
\begin{align}
u_k(\bU_k) =&\ {\rm Tr}\left({\bf U}\T_k \bTheta^{(\calP)}_{k|k} {\bf U}_k \right), \label{eq:utility_gain_u} \\
\ell_{k+n|k}(\bU_k) =&\ \calF \left(\left[ {\rm Tr}\left({\bf U}\T_k \bTheta^{(\calQ_1)}_{k+n|k} {\bf U}_k \right),\ldots,  {\rm Tr}\left({\bf U}\T_k \bTheta^{(\calQ_{|\calQ|})}_{k+n|k} {\bf U}_k \right)\right]^T\right). \label{eq:privacy_loss_u}
\end{align}
where we let $\calQ_i=\{|\calP|+i\}$ be a set consisting of the index of the $i$-th private state such that $\calQ=\bigcup\limits_{i=1}^{|\calQ|}\calQ_i$.

Recall that $\calF(\bx)=\bA\bx$. The Lagrangian $L({\bf U}_k,\{\gamma_{n,f}\},\{\lambda_m\})$ of problem \cref{eq:multi-step_problem_v2} is then given by 
\begin{align}
L({\bf U}_k,\{\gamma_{n,f}\},\{\lambda_m\})=&
\sum_{m=1}^{M} 
\left[{\bf U}_k\right]\T_{:,m}
\bTheta^{(\calP)}_{k|k}
\left[{\bf U}_k\right]_{:,m} \nonumber\\
&-\sum_{n=0}^{r_k} \sum_{f=1}^{|\calF|}
\gamma_{n,f}\left(\sum_{m=1}^{M} 
\left[{\bf U}_k\right]\T_{:,m}
\sum_{j=1}^{|\calQ|} \bA(f,j)\bTheta^{(\calQ_j)}_{k+n|k}
\left[{\bf U}_k\right]_{:,m}-\right. \nonumber\\
&\left.\left[\bA\right]_{\calI_f}{\bar \bdelta}_{k+n|k}
\right) -\sum_{m=1}^{M} \lambda_m\left(\left[{\bf U}_k\right]\T_{:,m} \left[{\bf U}_k\right]_{:,m} -1\right), \label{eq:lagr_function}
\end{align}
where $\gamma_{n,f}$, $n=0,\ldots,r_k$, $f=1,\ldots,|\calF|$, $\calI_f=\{f\}$ and $\lambda_m$, $m=1,\ldots,M$, are the Lagrange multipliers. 
Differentiating $L({\bf U}_k,\{\gamma_{n,f}\},\{\lambda_m\})$ with respect to $\left[{\bf U}_k\right]_{:,m}$ and equating to zero  leads to
\begin{align}
&\left(
\bTheta^{(\calP)}_{k|k}-\sum_{n=0}^{r_k} \sum_{f=1}^{|\calF|} \gamma_{n,f}\sum_{j=1}^{|\calQ|} \bA(f,j)\bTheta^{(\calQ_j)}_{k+n|k}
\right)\left[{\bf U}_k\right]_{:,m}  =
\lambda_m \left[{\bf U}_k\right]_{:,m}, \ m=1,\ldots,M,\label{eq:kkt_stationary_v2}
\end{align}
which implies the objective of \cref{eq:multi-step_problem_v2} is maximized when ${\bf U}_k$ consists of the $M$ unit eigenvectors of $\bTheta^{(\calP)}_{k|k}- \sum_{n=0}^{r_k}\sum_{f=1}^{|\calF|} \gamma_{n,f} \sum_{j=1}^{|\calQ|} \bA(f,j)\bTheta^{(\calQ_j)}_{k+n|k}$ associated with its $M$ largest \emph{non-zero} eigenvalues. Observe that there is no need to consider zero eigenvalues in our solution as these do not change the Lagrangian in \cref{eq:lagr_function}. We see that both $\{{\bf U}_{k}\}$ and $\{\lambda_m\}$ depend on $\{\gamma_{n,f} : n=0,\ldots,r_k;~f=1,\ldots,|\calF|\}$ and $M$. 

Problem \cref{eq:multi-step_problem_v2} can be equivalently recast as
\begin{align}\label{eq:primal} 
\underset{\{\gamma_{n,f}\}}{\max}&~ u_k(\bU_k)\\
{\rm s.t. }
&~{\rm Tr}\left({\bf U}\T_k \sum_{j=1}^{|\calQ|} \bA(f,j)\bTheta^{(\calQ_j)}_{k+n|k} {\bf U}_k \right)\leq \left[\bA\right]_{\calI_f}{\bar \bdelta}_{k+n|k},\nonumber\\
&\ n=0,\ldots,r_k,~f=1,\ldots,|\calF|, \nonumber\\
&~{\bf U}_k=\eig_{\ne0}\left(\bTheta^{(\calP)}_{k|k}- \sum_{n=0}^{r_k}\sum_{f=1}^{|\calF|} \gamma_{n,f} \sum_{j=1}^{|\calQ|} \bA(f,j)\bTheta^{(\calQ_j)}_{k+n|k},M\right), \nonumber
\end{align}
where $\eig_{\ne0}\left({\bf A},M\right)$ denotes the matrix of eigenvectors of $\bA$ corresponding to its $M$ largest non-zero eigenvalues. The optimization problem \cref{eq:primal} can be solved  using standard iterative methods such as the interior-point method and sequential quadratic programming.
However, due to non-convexity, there is no guarantee of finding the global optimum using such iterative methods.
Comparing \cref{eq:primal} (or equivalently \cref{eq:multi-step_problem_v2}) with original problem \cref{eq:multi-step_problem}, we can see that \cref{eq:multi-step_problem} is optimizing over a matrix $\bC_k$ while \cref{eq:primal} is optimizing over scalars ${\{\gamma_{n,f}\}}$. The optimal $\bU_k$ (or equivalently $\bC_k$) is then given by a closed-form expression in terms of $\{\gamma_{n,f}\}$. Hence, solving \cref{eq:primal} is computationally easier than solving \cref{eq:multi-step_problem} directly.

%

Since $\sum_{j=1}^{|\calQ|} \bA(f,j)\bTheta^{(\calQ_j)}_{k+n|k}$ are positive definite for every $f=1,\ldots,|\calF|$, the privacy loss function $\ell_{k+n|k}({\bf U}_k)$ is increasing in $M$. Hence, the privacy constraint in \cref{eq:multi-step_problem_v2} may not be feasible when $M$ is large. To determine the optimal $M$, we have the following lemma.

\begin{Lemma} \label{lemma:M}
For two positive integers $M^+$ and $M^-$ such that $M^+>M^-$, suppose that $\{\gamma_{n,f}^+\}$ and $\{\gamma_{n,f}^-\}$ are the solutions of \cref{eq:primal} when $M=M^+$ and $M=M^-$, respectively, and let $\bU_k^+$ and $\bU_k^-$ be the corresponding $\bU_k$. Then, $u_k(\bU_k^+)\geq u_k(\bU_k^-)$.
\end{Lemma}
\begin{IEEEproof}
Let 
\begin{align*}
\bPhi_k^- &=\bTheta\tc{\calP}_{k|k}-\sum_{n=0}^{r_k}\sum_{f=1}^{|\calF|} \gamma_{n,f}^- \sum_{j=1}^{|\calQ|} \bA(f,j)\bTheta^{(\calQ_j)}_{k+n|k}, \\
\bPhi_k^+ &=\bTheta\tc{\calP}_{k|k}- \sum_{n=0}^{r_k}\sum_{f=1}^{|\calF|} \gamma_{n,f}^+ \sum_{j=1}^{|\calQ|} \bA(f,j)\bTheta^{(\calQ_j)}_{k+n|k}.
\end{align*} 
Then, we have 
\begin{align*}
u_k(\bU_k^-)&= {\rm Tr}\left(\eig_{\ne0}(\bPhi_k^-,M^-)\T\bTheta\tc{\calP}_{k|k} \eig_{\ne0}(\bPhi_k^-,M^-) \right)\\
&\leq
{\rm Tr}\left(\eig_{\ne0}(\bPhi_k^-,M^+)\T\bTheta\tc{\calP}_{k|k} \eig_{\ne0}(\bPhi_k^-,M^+) \right),
\\
&\leq
{\rm Tr}\left(\eig_{\ne0}(\bPhi_k^+,M^+)\T\bTheta\tc{\calP}_{k|k} \eig_{\ne0}(\bPhi_k^+,M^+) \right)\\
&=u_k({\bf U}_k^+),
\end{align*}
where the first inequality follows because $\bTheta\tc{\calP}_{k|k}\succeq {\bf 0},M^+>M^-$ and the last inequality follows because the maximum of $u_k({\bf U}_k^+)$ is achieved when $\gamma_{n,f}=\gamma_n^+$, for $n=0,\ldots, r_k,f=1,\ldots,|\calF|$. The lemma is now proved.
\end{IEEEproof}

\cref{lemma:M} shows that as long as the privacy constraints in \cref{eq:multi-step_problem_v2} are satisfied, $M$ should be chosen as large as possible to maximize $u_k$. Let $\bPhi_k=\bTheta^{(\calP)}_{k|k}- \sum_{n=0}^{r_k}\sum_{f=1}^{|\calF|}\gamma_{n,f} \sum_{j=1}^{|\calQ|} \bA(f,j)\bTheta^{(\calQ_j)}_{k+n|k}$. Due to the fact that ${\rm rank}(\bA+\bB)\leq{\rm rank}(\bA)+{\rm rank}(\bB)$, ${\rm rank}\left(\bPhi_k\right)\leq \min(N,|\calP|+(r_k+1)|\calQ|)$, thus we have $M\leq \min(N,|\calP|+(r_k+1)|\calQ|)$. The distribution of the eigenvalues of $\bPhi_k$ is revealed by the following lemma.

\begin{Lemma}\label{lemma:rank}
Let ${\bf A}\succeq {\bf 0}$ and ${\bf B}\succeq {\bf 0}$ be two $N\times N$ matrices whose ranks are $r_\bA$ and $r_\bB$, respectively. If $N>r_\bA+r_\bB$, $\bC=\bA-\bB$ has at most $r_\bA$ positive eigenvalues and at most $r_\bB$ negative eigenvalues.
\end{Lemma}
\begin{IEEEproof}
See Appendix~\ref{app:lemma:rank}
\end{IEEEproof}

Suppose the eigenvalues of $\bPhi_k$ are sorted in descending order as $\lambda_1,\ldots,\lambda_N$. \cref{lemma:rank} shows that if $N>|\calP|+(r_k+1)|\calQ|$, we should choose $\bU_k$ to consist of the $M$ unit eigenvectors of $\bPhi_k$ associated with its $M$ largest non-zero eigenvalues in $\{\lambda_1,\ldots,\lambda_{|\calP|},\lambda_{N-(r_k+1)|\calQ|+1},\ldots,\lambda_N\}$. If $N\gg |\calP|+(r_k+1)|\calQ|$, the measurement can be significantly compressed without losing utility.

We summarize our solution approach to \cref{eq:multi-step_problem} in \cref{alg:central}.

\begin{algorithm}
\caption{Centralized solver for \cref{eq:multi-step_problem}}
\label{alg:central}
\begin{algorithmic}[1]
\REQUIRE 
$\calF(\cdot), r_k, \delta, \tilde\bP_{k-1|k-1}, {\bf H}_k, {\bf R}_k, {\bf F}_{k+n}, {\bf Q}_{k+n}, n=0,\ldots,r_k$
\ENSURE $M,\bC_k={\bf U}_k\T\bT_k^{-1/2}$ 
\STATE  Compute $\{\bTheta_{k+n|k}\}$ according to \cref{bTheta}, and $\{\calF\left(\bar{\bdelta}_{k+n}\right)\}$ defined in \cref{eq:multi-step_problem_v2}
\STATE Initialize $M=\min(\{N,|\calP|+(r_k+1)|\calQ|\}$
\WHILE{$M\geq 1$}
\STATE For each $M$, find $\{\gamma_{n,f}\}$ by solving \cref{eq:primal}.
\STATE Compute $\bU_k$ by
\begin{align*}
\eig_{\neq 0}\left(\bTheta^{(\calP)}_{k|k} - \sum_{n=0}^{r_k}\sum_{f=1}^{|\calF|}\gamma_{n,f} \sum_{j=1}^{|\calQ|} \bA(f,j)\bTheta^{(\calQ_j)}_{k+n|k},M\right).
\end{align*}
\IF{$\ell_{k+n|k}(\bU_k)\leq \calF\left(\bar{\bdelta}_{k+n}\right),\ n=0,\ldots,r_k$}
\STATE break
\ENDIF
\STATE $M \leftarrow M-1$
\ENDWHILE
\end{algorithmic}
\end{algorithm}

\section{Decentralized Sensors}
\label{section:dist}
Suppose that there are $S>1$ sensors and each sensor $s$ makes a $N_s\times 1$ measurement. Dividing the measurement model in \cref{eq:obs} into $S$ parts, the measurement model at each sensor $s$ is given by 
\begin{align}
\left[\bz_k\right]_{{\cal I}_s}&=\left[{\bf H}_k\right]_{{\cal I}_s} {\bf x}_{k} + \left[{\bf n}_k\right]_{{\cal I}_s}, ~s=1,\ldots,S,\label{eq:obs_individual}
\end{align}
where ${\cal I}_s=\{1+\sum_{i=1}^{s-1}N_i,\ldots,\sum_{i=1}^{s}N_i\}$ and $\sum_{s=1}^{S}N_s=N$. The quantities $\left[\bz_k\right]_{{\cal I}_s}$, $\left[{\bf H}_k\right]_{{\cal I}_s}\in \Real^{N_s\times L}$ and $\left[{\bf n}_k\right]_{{\cal I}_s}\in \Real^{N_s}$ are, respectively, the measurement made by sensor $s$, the measurement matrix of sensor $s$, and the measurement noise. 

Each sensor $s$ applies its own compression $\pC_\JsIs \in\Real^{M_s\times N_s}$, which is a linear mapping from a vector space with dimension $N_s$ to one with dimension $M_s<N_s$, and $\calJ_s=\{1+\sum_{i=1}^{s-1}M_i,\ldots,\sum_{i=1}^{s}M_i\}$, on its measurement before sending it to the fusion center. We have
\begin{align}
\left[\tbz_k\right]_{{\cal I}_s}&= \pC_{\calJ_s,\calI_s} \cdot \left[\bz_k\right]_{{\cal I}_s} \nonumber\\
&= 
\pC_\JsIs \left[{\bf H}_k\right]_{{\cal I}_s} {\bf x}_{k} + 
\pC_\JsIs \left[{\bf n}_k\right]_{{\cal I}_s}. \label{eq:obs_individual_compressed}
\end{align}
Relating the compression map at each sensor in \cref{eq:obs_individual_compressed} to the overall compression map in \cref{eq:obs_compressed}, we see that the distributed implementation restricts the structure of the transformation matrix $\bC_{k}$ to be a block diagonal matrix whose diagonal entries are the transformation matrices $\pC_\JsIs,s=1,\ldots,S$ and the remaining entries are all zeros. The version of \cref{eq:multi-step_problem_v2} for decentralized sensors can then be formulated as follows: 
\begin{align}\label{eq:dist_multi-step_problem}\tag{P3}
\underset{\bC_k}{\max} &~  
u_k(\bC_k)\\
{\rm s.t.}
&~\ell_{k+n|k}(\bC_k)\leq \calF\left(\bar{\bdelta}_{k+n}\right),\ n=0,\ldots,r_k, \nonumber \\
&~\bC_{k}={\rm diag}\left(\pC_\JlIl,\ldots,\pC_\JSIS\right). \nonumber
\end{align}
 We firstly rewrite \cref{D} to express $\bD_{k+n|k}$ in terms of $\pC_\JlIl,\ldots,\pC_\JSIS$ as follows:
\begin{align}
\bD_{k+n|k} =&\ {\bf G}_{k+n|k}\T \bC\T_k \left(
\bC_k\bT_k\bC\T_{k+n|k}
\right)^{-1} \bC_k {\bf G}_{k+n|k},  \nonumber\\
=&\ 
\left[\pG_{{\cal I}_1}\T\pC\T_\JlIl,\ldots,
\pG_{{\cal I}_S}\T\pC\T_\JSIS\right] \nonumber\\
&\
\bPsi^{-1}_k
\left[\begin{array}{c}
\pC_\JlIl \pG_{{\cal I}_1} \\
\vdots\\
\pC_\JSIS \pG_{{\cal I}_S}
\end{array}
\right],
\label{eq:p_k_tilde}
\end{align}
where $\bPsi_k$ is defined in \cref{Psik}.
\begin{figure*}[!t]
\begin{align}
\bPsi_k &= 
\left[
\begin{array}{ccc}
\pC_\JlIl \pT_\IlIl \pC\T_\JlIl & \ldots & \pC_\JlIl \pT_\IlIS \pC\T_\JSIS \\
\vdots & \ddots & \vdots\\
\pC_\JSIS \pT_\ISIl \pC\T_\JlIl & \ldots & \pC_\JSIS \pT_\JSIS \pC\T_\JSIS
\end{array}
\right].\label{Psik}
\end{align}
\end{figure*}

To solve \cref{eq:dist_multi-step_problem} in a decentralized fashion, we aim to separate the objective into individual local objective functions at each sensor. However, without additional assumptions and reformulation of the objective, this is not possible due to the inverse of the matrix $\bPsi_k$. In the following, we consider two cases: 1) where there is no information exchange between sensors; and 2) where sensors are allowed to broadcast messages sequentially to all other sensors. In both these cases, we propose new objective functions that are separable.

\subsection{With no information exchange between sensors}
We assume that each sensor at time $k$ only knows its own measurement model, i.e., $\left[{\bf H}_k\right]_{{\cal I}_s}$, and the covariance matrix $\left[{\bf R}_k\right]_{{\cal I}_s,{\cal I}_s}$ of the process noise $\left[{\bf n}_k\right]_{{\cal I}_s}$ at sensor $s$. To make the objective function separable, we propose to ignore all the inter-sensor terms, i.e., $\pT_{{\cal I}_i,{\cal I}_j},i\neq j$, and replace $\bPsi_k$ in \cref{Psik} with the following approximation
\begin{align}
{\rm diag}\left(\pC_{{\cal I}_1,\calJ_1}\pT_{{\cal I}_1,{\cal I}_1}\pC\T_{{\cal I}_1,\calJ_1},\ldots, \pC_{{\cal I}_S,\calJ_S}\pT_{{\cal I}_S,{\cal I}_S}\pC\T_{{\cal I}_S,\calJ_S}
\right),\label{eq:psi_no_message}
\end{align}
in which only the diagonal terms containing $\pT_{{\cal I}_s,{\cal I}_s},s=1,\ldots,S$ are retained. 
Plugging \cref{eq:psi_no_message} in place of $\bPsi_k$ into \cref{eq:p_k_tilde} and using \cref{lemma:U=A} yields
\begin{align}
\bD_{k+n|k} &\approx
\sum_{s=1}^S 
\left[\bG_{k+n|k}\right]_{\calI_s}\T\pC\T_\JsIs \left(\pC_\JsIs\pT_\IsIs\pC\T_\JsIs\right)^{-1} \pC_\JsIs\pG_{\calI_s},\nonumber\\
&=
\sum_{s=1}^S \pG\T_{\calI_s}\pT_\IsIs^{-1/2}{\bf U}_k^{(s)} {{\bf U}_k^{(s)}}\T \pT_\IsIs^{-1/2}\pG_{\calI_s},
\label{eq:p_k_plus_n_given_k_no_msg}
\end{align}
where $\bU_k\tc{s}$ consists of the $M_s$ right singular vectors associated with the $M_s$ non-zero singular values of $\pC_\JsIs\pT_\IsIs^{1/2}$ and ${{\bf U}_k^{(s)}}\T{{\bf U}_k^{(s)}}=\bI_{M_s}$. From \cref{eq:p_k_plus_n_given_k_no_msg}, we can now write the utility gain and privacy loss in \cref{eq:dist_multi-step_problem} approximately as
\begin{align}
u_k(\bC_k) &\approx \sum_{s=1}^S u_k^{(s)}(\bU_k\tc{s}),\label{g_approx}\\
\ell_{k+n|k}(\bC_k) &\approx \sum_{s=1}^S \ell_{k+n|k}^{(s)}(\bU_k\tc{s}), \label{ell_approx}
\end{align}
where
\begin{align}
u_{k}^{(s)}({\bf U}^{(s)}_k) &= {\rm Tr}\left({{\bf U}^{(s)}_k}\T \bTheta^{({\cal P},s)}_{k|k} {\bf U}^{(s)}_k \right),\label{eq:utility_gain_s}\\
\ell_{k+n|k}^{(s)}({\bf U}^{(s)}_k) &= 
\calF \left(\left[ {\rm Tr}\left({{\bf U}^{(s)}_k}\T \bTheta^{(\calQ_1,s)}_{k+n|k} {\bf U}^{(s)}_k \right),\ldots,{\rm Tr}\left({{\bf U}^{(s)}_k}\T \bTheta^{(\calQ_{|\calQ|},s)}_{k+n|k} {\bf U}^{(s)}_k \right)\right]^T\right)
,\label{eq:privacy_loss_kpn_s}
\end{align}
and
\begin{align*}
\bTheta^{(\calI,s)}_{k|k}&=\pT_\IsIs^{-1/2}\left[\bG_{k|k}\right]_{\calI_s, \calI} \left[\bG_{k|k}\right]\T_{\calI_s, \calI}\pT_\IsIs^{-1/2}
\end{align*}
for any index set $\calI$.
The local ${\bf U}^{(s)}_k$, for each $s=1,\ldots,S$ can be optimized separately by solving the following problem at each sensor $s$ using a procedure similiar to \cref{alg:central}:
\begin{align}\label{eq:dist_no_message_multi-step_problem}\tag{P3.1} 
\max_{{\bf U}^{(s)}_k}&\ 
u_k^{(s)}({\bf U}^{(s)}_k) \nonumber\\
{\rm s.t. }
&\ \ell_{k+n|k}^{(s)}({\bf U}^{(s)}_k) \leq \calF\left( {\bar \bdelta}^{(s)}_{k+n}\right),\ n=0,\ldots,r_k, \nonumber\\
&\ {{\bf U}^{(s)}_k}\T {\bf U}^{(s)}_k={\bf I}_{M_s}, \nonumber
\end{align}
where 
\begin{align*}
{\bar \bdelta}^{(s)}_{k+n}={\rm vecdiag}\left(\left[\tilde\bP^{(s)}_{k+n|k-1}\right]_{\cal Q,Q} \right) - \delta\tc{s} {\bf 1}_{|\calQ|}
\end{align*}
and  
\begin{align*}
\tbP^{(s)}_{k+n|k-1}
=&{\bf F}_{k+1:k+n}\tbP^{(s)}_{k|k-1}\bF\T_{k+1:k+n}  + \sum_{i=1}^{n} {\bf F}_{k+i+1:k+n}\bQ_{k+i}\bF\T_{k+i+1:k+n}
\end{align*}
with $\tbP^{(s)}_{k|k-1}$ computed based on local information only. Since we assume that sensors do not know each other's measurement statistics, they cannot coordinate amongst themselves to choose ${\delta}\tc{s}$. For simplicity, we choose $\delta/S\leq \delta\tc{s} < \min \calF\left({\rm vecdiag}\left(\left[\tilde\bP^{(s)}_{k+n|k-1}\right]_{\cal Q,Q} \right)\right)$.  
Since \cref{ell_approx} is an approximation, there is no guarantee that solving \cref{eq:dist_no_message_multi-step_problem} at every sensor produces a global feasible solution. This is mainly due to the lack of information exchange. However, this scheme can be used to initialize a more sophisticated iterative scheme that we introduce in next subsection.

\subsection{Sequential message broadcasts}\label{subsec:info_exchange}
In the formulation in the previous subsection, $\bPsi_k$ is approximated as a block diagonal matrix in \cref{eq:psi_no_message}. However, the off-diagonal/inter-sensor terms in $\bPsi_k$ may not be negligible, and ignoring them may compromise the privacy-utility tradeoff. In this subsection, we consider the case where information is exchanged between sensors to facilitate optimization of the compression map at each sensor.

We rewrite $\bD_{k+n|k} $ in \cref{eq:p_k_tilde} to isolate sensor $s$'s compression map, $\pC_\IsJs$, from the transformations of the other sensors $\pC_\JnsIns$ (see Appendix~\ref{A2} for the derivation):  
\begin{align}
\bD_{k+n|k} =&\pG_{\calI_{\s}}\T \pC\T_\JnsIns \left[\bPhi_k\right]_\JnsJns^{-1} \pC_\JnsIns \pG_{\calI_{\s}} \nonumber\\
&+ \left(\bG_{k+n|k}^{(s)}\right)\T {\bf U}^{(s)}_k {{\bf U}^{(s)}_k}\T \bG_{k+n|k}^{(s)}-\left(\bG_{k+n|k}^{(s)}\right)\T {\bf U}^{(s)}_k {{\bf U}^{(s)}_k}\T \bG_{k+n|k}^{(\s)}  \nonumber\\
&-\left(\bG_{k+n|k}^{(\s)}\right)\T {\bf U}^{(s)}_k {{\bf U}^{(s)}_k}\T \bG_{k+n|k}^{(s)} + \left(\bG_{k+n|k}^{(\s)}\right)\T {\bf U}^{(s)}_k {{\bf U}^{(s)}_k}\T \bG_{k+n|k}^{(\s)} ,\label{eq:pkplusn_info_ex_us}
\end{align}
where $\bU_k\tc{s}$ consists of the $M_s$ right unit singular vectors of $\pC_\IsJs \left(\bT_k^{(s)}\right)^{1/2}$ associated with its non-zero singular values and 
\begin{align*}
\bT_k^{(s)} &=\pT_\IsIs - \pT_\IsIns \pC\T_\JnsIns \left[\bPhi_k\right]_\JnsJns^{-1} \pC_\JnsIns\pT_\IsIns\T, \\
\bG_{k+n|k}^{(s)}&=\left(\bT_k^{(s)}\right)^{-1/2}\pG_{\calI_{s}},\\
\bG_{k+n|k}^{(\s)}&=\left(\bT_k^{(s)}\right)^{-1/2} \pT_\IsIns \pC\T_\JnsIns \left[\bPhi_k\right]_\JnsJns^{-1} \pC_\JnsIns \pG_{\calI_{\s}},\\
\calI_{\s} &=\calI_{1}\cup \ldots \cup \calI_{s-1}\cup\calI_{s+1}\cup\ldots\cup\calI_{S}.
\end{align*}

For any $n\geq 0$ and any index set $\calI$, let 
\begin{align*}
{\bXi}\tc{\calI,s}_{k+n|k} =&
  \left[\bG_{k+n|k}^{(s)}\right]_{:, \calI}  \left[\bG_{k+n|k}^{(s)}\right]_{:, \calI}\T
- \left[\bG_{k+n|k}^{(\s)}\right]_{:, \calI} \left[\bG_{k+n|k}^{(s)}\right]_{:, \calI}\T \\
&-\left[\bG_{k+n|k}^{(s)}\right]_{:, \calI}  \left[\bG_{k+n|k}^{(\s)}\right]_{:, \calI}\T
+ \left[\bG_{k+n|k}^{(\s)}\right]_{:, \calI} \left[\bG_{k+n|k}^{(\s)}\right]_{:, \calI}\T. 
\end{align*} 
Then, the utility $u_k$ and privacy loss $\ell_{k+n|k}$ for $n\geq 0$ can be rewritten as 
\begin{align*}
u_k =& {\rm Tr}\left(\left[\bD_{k+n|k}\right]_{\cal P,P}\right) \\
=& {\rm Tr}\left({{\bf U}^{(s)}_k}\T {\bXi}\tc{\calP,s}_{k|k} {\bf U}^{(s)}_k \right)+
{\rm Tr}\left(\left[\pG_{\calI_{\s}}\T \pC\T_\JnsIns \right.\right. \\ 
&\left.\left.\left[\bPhi_k\right]_\JnsJns^{-1} 
\pC_\JnsIns \pG_{\calI_{\s}}\right]_{\cal P,P}\right),
\end{align*}
and
\begin{align*}
&\ell_{k+n|k} \nonumber\\
=& \calF\left( {\rm vecdiag}\left(\left[\bD_{k+n|k}\right]_{\cal Q,Q}\right) \right) \\
=& \calF\left(\left[{\rm Tr}\left({{\bf U}^{(s)}_k}\T {\bXi}\tc{\calQ_1,s}_{k+n|k} {\bf U}^{(s)}_k \right),\ldots,{\rm Tr}\left({{\bf U}^{(s)}_k}\T {\bXi}\tc{\calQ_{|\calQ|},s}_{k+n|k} {\bf U}^{(s)}_k \right)\right]^T \right)\\
&+
\calF\left({\rm vecdiag}\left(\left[\pG_{\calI_{\s}}\T \pC\T_\JnsIns \left[\bPhi_k\right]_\JnsJns^{-1} 
\pC_\JnsIns \pG_{\calI_{\s}}\right]_{\cal Q,Q}\right)\right).
\end{align*}


Since $u_k$ and $\ell_{k+n|k}$ depend on some information that are not available locally at each sensor, the following information need to be shared between the sensors: 
\begin{enumerate}
\item $\pC_\JsIs$, the transformations applied locally at sensor $s$;
\item $\left[\bH_{k}\right]_{\calI_{s}}$, the local measurement matrix of sensor $s$;
\item $\left[\bR_k\right]_\IsJs$, the covariance matrix of the local measurement noise at sensor $s$;
\item $\left[\bR_k\right]_\IsJns$, the correlation between the local measurement noise at sensor $s$ and that of other sensors. 
\end{enumerate} 
The first item is used to construct $\pC_\JnsIns$ while the last three items are needed to compute $\bG_{k+n|k}^{(\s)}$ at sensor $s$. However, the correlations between the noise measured at different sensors maybe difficult to know \emph{a priori} in practice if the measurement are made separately. In such a case, we assume the noise measured at different sensors are independent, i.e., $\left[{\bf R}_k\right]_{\calI_{i},\calJ_{j}}={\bf 0},\forall i\neq j$.

To solve \cref{eq:dist_multi-step_problem}, we propose an alternating optimization procedure. 
We consider a sequential message passing schedule: each sensor in turn transmits messages to all other sensors. Suppose the order of transmission is predefined as $1,2,\cdots,S$. Sensor $s$ finds $\pC_\JsIs^{\{i\}}$ at iteration $i$ using $\left[\pC^{\{i\}}_\JlIl,\ldots,\pC_{\calJ_{s-1},\calI_{s-1}}^{\{i\}}, \pC_{\calJ_{s+1},\calI_{s+1}}^{\{i-1\}},\ldots,\pC_\JSIS^{\{i-1\}}\right]$ as $\pC_\JnsIns$. The details are summarized in \cref{alg:decentralized}. 

\begin{Proposition}\label{prop:convergence}
In \cref{alg:decentralized}, suppose that $u^{\{i\}}_k$ is the utility gain at iteration $i$. Then $u^{\{i\}}_k$ converges under the sequential schedule.
\end{Proposition}
\begin{proof}
For each sensor $s$ and iteration $i$, let $u_k^{(s)}(\pC_{\calJ_{1},\calI_{1}}^{\{i\}},\ldots,\pC_\JsIs^{\{i\}},\pC_{\calJ_{s+1},\calI_{s+1}}^{\{i-1\}},\pC_\JSIS^{\{i-1\}})$ be the utility at sensor $s$ in iteration $i$ under the sequential schedule. We have
\begin{align*}
u^{\{i-1\}}_k =& u^{(S)}_k(\pC_\JlIl^{\{i-1\}}, \ldots, \pC_\JSIS^{\{i-1\}}), ~\mbox{iteration $i-1$ at sensor $S$}\\
\leq&
u^{(1)}_k(\pC_{\calJ_{1},\calI_{1}}^{\{i\}},\pC_{\calJ_{2},\calI_{2}}^{\{i-1\}},\ldots,\pC_\JSIS^{\{i-1\}}), ~\mbox{iteration $i$ at sensor $1$} \\
&\vdots\\
\leq&
u^{(s)}_k(\pC_{\calJ_{1},\calI_{1}}^{\{i\}},\ldots,\pC_\JsIs^{\{i\}},\pC_{\calJ_{s+1},\calI_{s+1}}^{\{i-1\}},\pC_\JSIS^{\{i-1\}}), ~\mbox{iteration $i$ at sensor $s$} \\
&\vdots\\
\leq&
u^{(S)}_k(\pC_{\calJ_{1},\calI_{1}}^{\{i\}},\ldots,\pC_\JSIS^{\{i\}}),~\mbox{iteration $i$ at sensor $S$} \\
=& u^{\{i\}}_k,
\end{align*}
where the above inequalities follow because for each sensor $s$, $\pC_\JsIs^{\{i\}}$ is obtained by maximizing the objective of \cref{eq:dist_multi-step_problem} given 
$\pC_{\calJ_{1},\calI_{1}}^{\{i\}},\ldots,
\pC_{\calJ_{s-1},\calI_{s-1}}^{\{i\}},
\pC_{\calJ_{s+1},\calI_{s+1}}^{\{i-1\}},\ldots
\pC_\JSIS^{\{i-1\}}$. Since $u^{\{i\}}_k \leq \trace{\submat{\tbP_{k|k-1}}{\calP,\calP}}$ for all $i\geq 1$, the proposition is proved.
\end{proof}

\begin{algorithm}
\caption{Distributed algorithm: sequential update schedule at time step $k$}
\label{alg:decentralized}
\begin{algorithmic}[1]
\REQUIRE $\pC_\JsIs^{\{0\}}={\bU_{k}^{(s)}}\T \left(\bT_k^{(s)}\right)^{-1/2},\forall s$ with $M_s^{\{0\}},{\bf U}_{k}^{(s)}$ obtained by solving \cref{eq:dist_no_message_multi-step_problem}
\ENSURE $\pC_\JsIs^{\{\rm NumIter\}},s=1,\ldots,S$
\STATE Num-iter $\leftarrow$ 1 \COMMENT{iteration index} 
\WHILE{$|u_k^{\{\rm NumIter\}}-u_k^{\{\rm NumIter - 1\}}|\geq\epsilon$}  
\FOR{$s=1$ \TO $S$ } 
\STATE Broadcast $\pC_\JsIs^{\{i-1\}}$, 
\STATE Broadcast $\left\{\left[{\bf H}_k\right]_{\calI_{s}},\left[{\bf R}_k\right]_\IsIs\right\}$ once only at $i=1$
\STATE Receive $\pC_\JnsIns^{\{i-1\}}$,
\STATE Receive $\left\{\left[{\bf H}_k\right]_{\calI_{\s}},\left[{\bf R}_k\right]_\InsIns\right\}$ once only at $i=1$ 
\STATE Update $\pC_\JsIs^{\{i\}}$ by solving \cref{eq:dist_multi-step_problem} where 
\begin{align*}
\pC_\JnsIns=&
\left[\pC_{\calJ_{1},\calI_{1}}^{\{i\}},\ldots,\pC_{\calJ_{s-1},\calI_{s-1}}^{\{i\}}, \pC_{\calJ_{s+1},\calI_{s+1}}^{\{i-1\}},\ldots,\pC_\JSIS^{\{i-1\}}\right].
\end{align*}
\ENDFOR
\STATE NumIter $\leftarrow$ NumIter $+1$
\ENDWHILE
\end{algorithmic}
\end{algorithm}


\cref{fig:pf} summarizes the relationships between the problem formulations \cref{eq:1step_problem}-\cref{eq:dist_multi-step_problem} in \cref{section:central} and \cref{section:dist}.
\begin{figure}[!htb]
\centering
    \includegraphics[width=.9\linewidth]{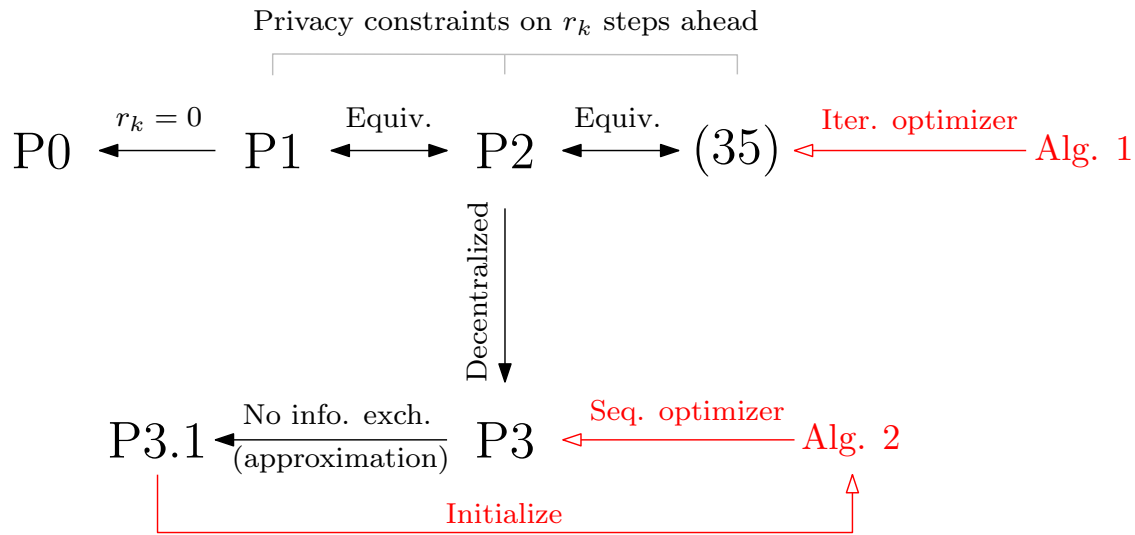}
  \caption{Relationships between the problem formulations \cref{eq:1step_problem}-\cref{eq:dist_multi-step_problem}.}
  \label{fig:pf}
\end{figure}

\section{Simulation Results}
\label{section:simulations}

In this section, we present simulation studies to understand the impact of different parameters on the utility-privacy tradeoff in both the centralized (single sensor) and distributed cases. We compare our compressive privacy scheme with the IB \cite{CheNIPS2003}, PF \cite{MakITW2014}, and CP \cite{KunSPM2017} privacy mechanisms. We use the following settings for all the simulations in this section (unless otherwise stated): the entries of ${\bf H}_k$ at each time step are drawn independently from ${\cal N}(0,1)$, ${\bf R}_k={\bf I}_N$, $\bP_{0|0}=0.01{\bf I}_L$, $|\calP|=4$, and $|\calQ|=4$. Each data point shown in the following figures is averaged over 50 independent experiments.

\subsection{Centralized case}
\label{sect:sim_central}

We first consider the centralized case where there is a single sensor. Recall that $\tau_k$ and $\eta_k$, which are defined in \cref{eq:utility_func} and \cref{eq:privacy_func}, denote the public error trace and private error function, respectively. The variable $r_k$ determines the number of future time steps that are considered in \cref{eq:multi-step_problem}, and when $r_k=0$, only the current time step is considered. We choose \cref{eq:eqk1} to be the private error function for \cref{fig:d1} and \cref{fig:w=02} so that $|\calF|=1$ and $\calF(\delta\bm{1}_{|\calQ|})=|\calQ|\delta$. We let \cref{eq:eqk2} to be the private error function for \cref{fig:d1_F} so that $|\calF|=|\calQ|$ and $\calF(\delta\bm{1}_{|\calQ|})=\delta\bm{1}_{|\calQ|}$.%

\begin{figure}[!htb]
\centering
    \includegraphics[width=.6\linewidth]{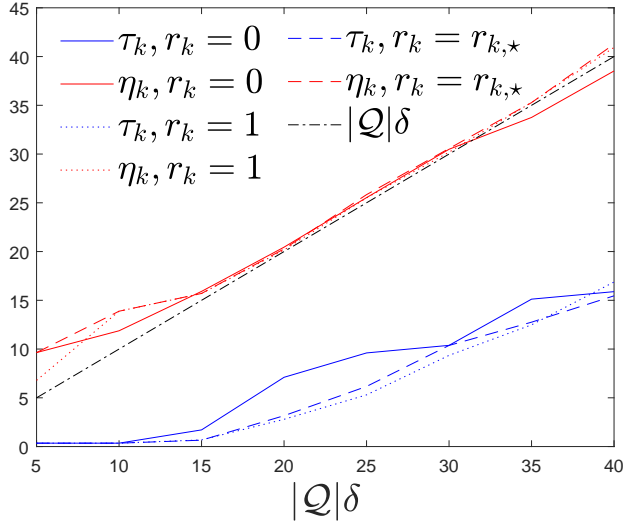}
  \caption{The public error trace $\tau_k$ and privacy error trace $\eta_k$ vs.\ the privacy threshold $|\calQ|\delta$ at $k=20$.}\label{fig:d1}
\end{figure}

\begin{figure}[!htb]
\centering
    \includegraphics[width=.6\linewidth]{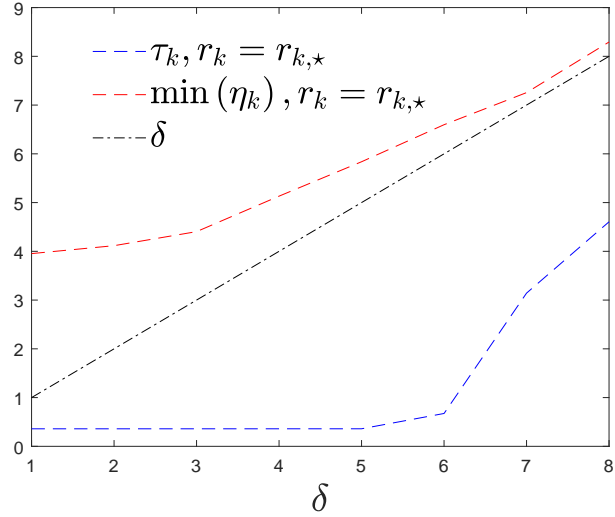}
  \caption{The public error trace $\tau_k$ and the minimum of private error variances $\min(\eta_k)$ vs.\ the privacy threshold $\delta$ at $k=20$ where the element-wise privacy constraint is considered and ${\bf Q}_k=5{\bf I}_N$.}\label{fig:d1_F}
\end{figure}

\cref{fig:d1} demonstrates the impact of $\delta$ on $\tau_k$ and $\eta_k$ with increasing $\delta$ and ${\bf Q}_k=2{\bf I}_N$. 
We randomly generate each $\bF_k$ for all $k\geq1$ by $\bF_k=\bU_k\bD_k\bV_k^{\T}$ where $\bU_k$ and $\bV_k$ contain the orthonormal bases of two randomly generated matrices, and $\bD_k$ is a diagonal matrix whose diagonal entries are uniformly drawn from $[1,1.2]$.  The results shown are at time step $k=20$. 

Recall in \cref{alg:central} that we keep reducing the compression dimension $M$ until  feasibility of the privacy constraint is achieved. In the case where $r_k$ is too small, the privacy constraint cannot be satisfied even when $M$ is reduced to 0, which implies that the information from previous time steps less than $k$ allows us to infer the private states at time $k$ better than the privacy constraint. In such a case, we set $M=0$ in our simulation result.  We see from \cref{fig:d1} that $\eta_k < |\calQ|\delta$ for $r_k=0$ when $\delta$ is sufficiently large.

In \cref{fig:d1,fig:d1_F}, we also show the case where $r_k=1$ and $r_k= r_{k,\star} = \underset{r\geq0}{\min}\ \left\{r: 2 r + \nu_k \geq \delta\right\}-1$ (see \cref{rkstar} of \cref{prop:min_n} where $\xi=1$ and $\epsilon=2$).  We see that in these cases, the privacy constraint $\eta_k\geq |\calQ|\delta$ is always satisfied. Moreover, we can also see in \cref{fig:d1,fig:d1_F} that the estimation of public states is clearly compromised as $\delta$ increases.

Recall that matrices ${\bf F}_{k+n}$, for $n=0,\ldots,r_k$,  play important roles in multi-step utility-privacy tradeoffs. To quantify the impact of ${\bf F}_{k+n}$, we let $\omega\in[0,1]$ and choose 
\begin{align*}
\bF_{k+n}&=
\left[
\begin{array}{cc}
\omega \left[{\bf F}\right]_{\cal P,P} & (1-\omega)\left[{\bf F}\right]_{\cal P,Q}\\
(1-\omega)\left[{\bf F}\right]_{\cal Q,P} & \omega \left[{\bf F}\right]_{\cal Q,Q}
\end{array}
\right],
\end{align*}
where the entries of $\bF$ for each $\bF_{k+n}$ are drawn independently from ${\cal N}(0,1)$, and each row of ${\bf F}$ is normalized to have unit norm. A smaller $\omega$ means that the public and private states are more correlated in the next time step. In \cref{fig:w=02}, we set $\omega=0.2$ to be small. When $r_k=0$, both public and private error traces evolve in a zigzag pattern and the privacy constraint $\eta_k\geq |\calQ|\delta$ is violated at every other time step. This is due to the high correlation between the public states $\submat{\bx_k}{\calP}$ at time $k$ and the private states $\submat{\bx_{k+1}}{\calQ}$ at time $k+1$. For example, a small public error trace $\tau_k$ at time $k$ yields a small $\trace{\submat{\tbP_{k+1|k}}{\calQ,\calQ}}$ at time $k+1$, which leads to a small private error trace $\eta_{k+1}$ at time $k+1$. On the other hand, if we set $r_k=1$, the additional privacy constraint $\eta_{k+1}\geq |\calQ|\delta$ ensures that $\eta_k\geq |\calQ|\delta$ is feasible from $k=2$ onwards. 

\begin{figure}[!htb]
\centering
    \includegraphics[width=.6\linewidth]{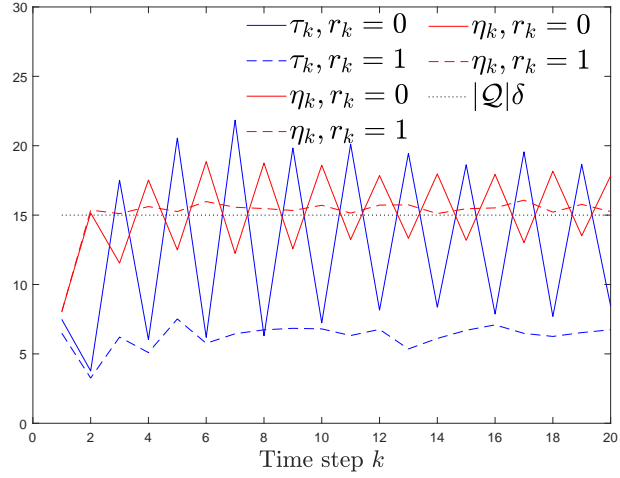}
\caption{The  public error trace $\tau_{k}$ and private error trace $\eta_{k}$ over 20 time steps where ${\bf Q}_k=2{\bf I}_N$, $|\calQ|\delta=15$ and $r_k=0,1$. }\label{fig:w=02}
\end{figure}


\subsection{Comparison amongst compressive privacy-preserving techniques}
\label{sect:cp}

If the state and process noises in \cref{system} are Gaussian random variables, then IB \cite{CheNIPS2003}, PF \cite{MakITW2014}, and CP \cite{KunSPM2017} can be regarded as compressive privacy-preserving techniques with different utility and privacy measures.  In this subsection, we briefly review these methods adapted to our problem formulation and compare them with our proposed approach,  where for every time $k$, we let $r_k=0$ and choose the private error trace \cref{eq:eqk1} as the privacy measure for our proposed approach. Since $|\calF|=1$, we write $\gamma_{n,1}$ in \cref{eq:lagr_function} as $\gamma_k$ to use the same symbol as the tradeoff parameters of the methods we compare with.

\begin{figure}[!htb]
\centering
\includegraphics[width=.6\linewidth]{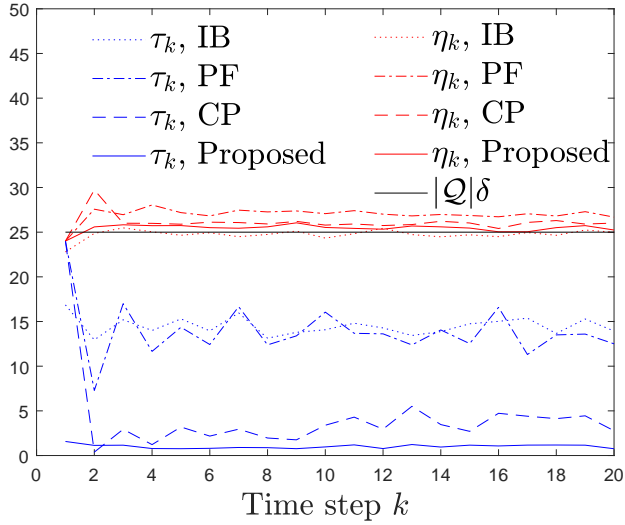}
\caption{Comparison of IB, PF, CP and proposed scheme in terms of $\tau_k$ while $\bH_k$ is a random matrix and $\eta_k=25$. We set $\bQ_k=4\bI_4$ and $r_k=0$. }\label{fig:rand}
\end{figure}

\begin{figure}[!htb]
\centering
\includegraphics[width=.6\linewidth]{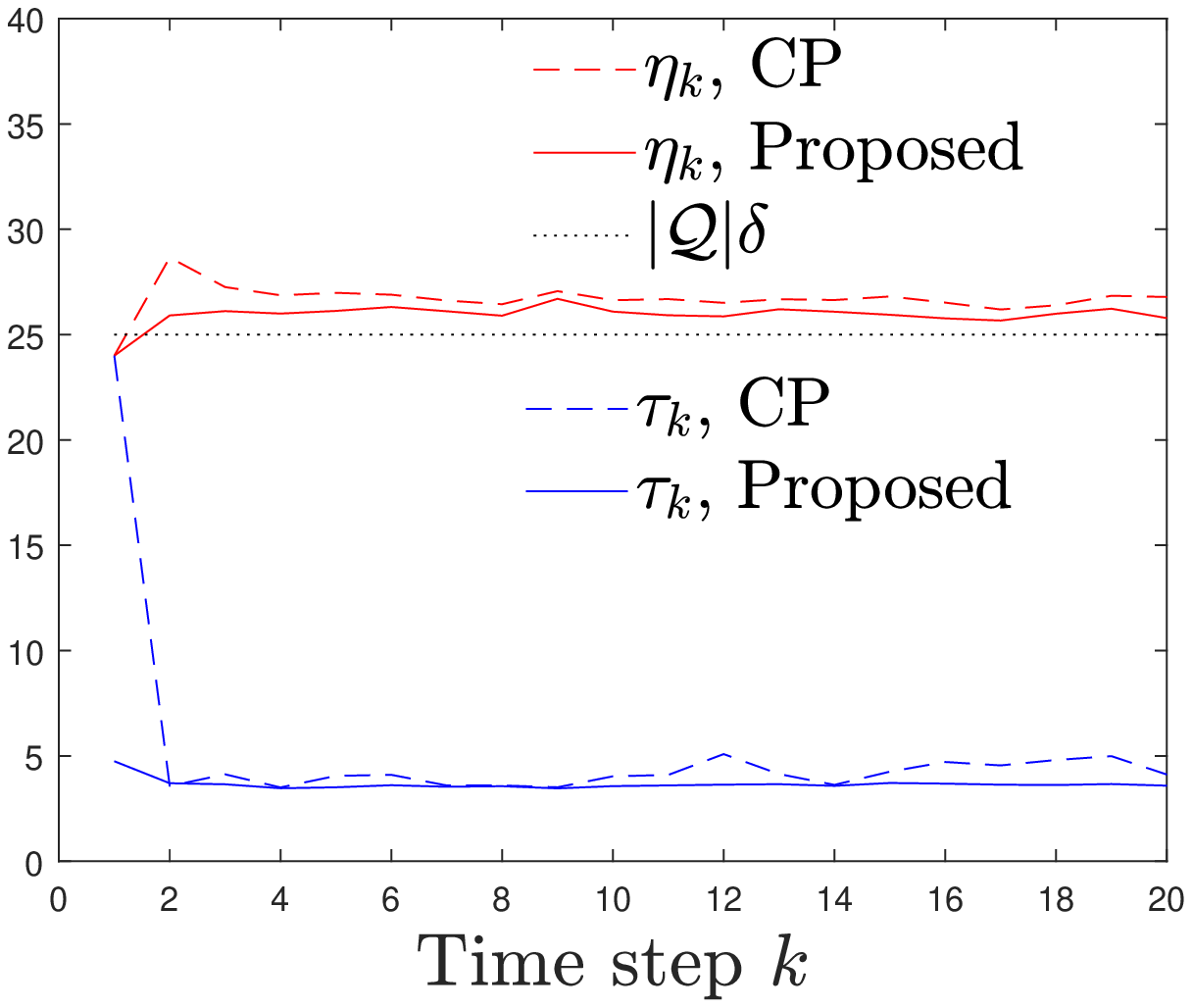}
\caption{Comparison of IB, PF, CP and proposed scheme in terms of $\tau_k$ while $\bH_k$ is an orthogonal matrix and $\eta_k=25$. We set $\bQ_k=4\bI_4$ and $r_k=0$. }\label{fig:orth}
\end{figure}

\begin{enumerate}[(i)]
\item IB finds the optimal compressive matrix $\bC_k$ by 
\begin{align*}
\underset{\bC_k}{\min} ~ I(\bz_k;\bC_k\bz_k)-\gamma_k I(\bC_k\bz_k;\by_k),
\end{align*}
where $\by_k = \left[\bx_k\right]_{\cal P}$ and $\gamma_k$ is a positive constant. The optimal compression is given by
\begin{align*}
\bC_k = \left[\alpha_1 \bv_1,\ldots,\alpha_M\bv_M \right]\T,
\end{align*}
where $\bv_i$, for $i=1,\ldots,M$ are the left eigenvectors of 
\begin{align*}
\left(\bT_{k} - \bSigma_{z_k,y_k} \left[\tbP_{k|k-1}\right]^{-1}_{\cal P, \cal P} \bT_k\right)\bT_{k}^{-1}
\end{align*}
with $\bSigma_{z_k,y_k} = \left[\bH_k\right]_{:,\cal P} \left[\tbP_{k|k-1}\right]_{\cal P,P}+\left[\bH_k\right]_{:,\cal Q} \left[\tbP_{k|k-1}\right]_{\cal Q,P}$, sorted by their corresponding ascending eigenvalues $\lambda_1,\ldots,\lambda_M$, and 
\begin{align*}
\alpha_i=\sqrt{\frac{\gamma_k(1-\lambda_i)-1}{\lambda_i \bv_i\T \bT_k \bv_i}}.
\end{align*}
 
\item PF finds the optimal compressive matrix $\bC_k$ by
\begin{align*}
\underset{\bC_k}{\max} ~ I(\bz_k;\bC_k\bz_k)-\gamma_k I(\bC_k\bz_k;\by_k),
\end{align*}
where $\by_k = \left[\bx_k\right]_{\cal Q}$ and $\gamma_k$ is a positive constant. The optimal compression is given by
\begin{align*}
\bC_k = \left[\alpha_1 \bv_1,\ldots,\alpha_M\bv_M \right]\T,
\end{align*}
where $\bv_i$, for $i=1,\ldots,M$ are the left eigenvectors of 
\begin{align*}
\left(\bT_{k} - \bSigma_{z_k,y_k} \left[\tbP_{k|k-1}\right]^{-1}_{\cal Q, \cal Q} \bT_k\right)\bT_{k}^{-1}
\end{align*}
with $\bSigma_{z_k,y_k} = \left[\bH_k\right]_{:,\cal P} \left[\tbP_{k|k-1}\right]_{\cal P,Q}+\left[\bH_k\right]_{:,\cal Q} \left[\tbP_{k|k-1}\right]_{\cal Q,Q}$, sorted by their corresponding descending eigenvalues $\lambda_1,\ldots,\lambda_M$.


\item CP finds the optimal  $\bC_k$ by 
\begin{align*}
\underset{\bC_k}{\max} ~ I(\bGamma^{(\calP)}_k\bz_k;\bC_k\bz_k)-\gamma_k I(\bGamma^{(\calQ)}_k\bz_k;\bC_k\bz_k),
\end{align*}
where $\bGamma^{(\calP)}_k=\left(\left[\bH_k\right]_{:,\cal Q}^\perp\right)\T$ and $\bGamma^{(\calQ)}_k=\left(\left[\bH_k\right]_{:,\cal P}^\perp\right)\T$. The optimal solution of $\bC_k$ is given by
\begin{align*}
\bC_k^* =\eig\left( \bOmega_k - \gamma_k \bPi_k, \bT_k , M \right),
\end{align*}
where ${\rm eig}\left({\bf A},{\bf B},M\right)$ is a matrix consisting of the $M$ principal generalized unit eigenvectors of the matrix pencil $\left({\bf A},{\bf B}\right)$, and 
\begin{align*}
\bOmega_k &= \bT_k \left(\bGamma^{(\calP)}_k\right)\T \left( \bGamma^{(\calP)}_k \bT_k \left(\bGamma^{(\calP)}_k\right)\T \right)^{-1} \bGamma^{(\calP)}_k \bT_k,\\
\bPi_k &= \bT_k \left(\bGamma^{(\calQ)}_k\right)\T \left( \bGamma^{(\calQ)}_k \bT_k \left(\bGamma^{(\calQ)}_k\right)\T \right)^{-1} \bGamma^{(\calQ)}_k \bT_k.
\end{align*}
\end{enumerate}


\cref{fig:rand} and \cref{fig:orth} compares the private error trace $\tau_k$ and private error trace $\eta_k$ of the different schemes IB, PF, CP, and our proposed approach where $r_k=0$. The tradeoff parameter $\gamma_k$ and the dimension $M$ are chosen to make $\eta_k$ equal to (or as close as possible to) $|\calQ|\delta=25$ when $k$ is large.
The entries of $\bF_k$ are drawn independently from ${\cal N}(0,1)$, and each row of ${\bf F}_k$ is normalized to have unit norm. 
In \cref{fig:rand} where $\bH_k$ is a random matrix, we see that our proposed scheme, whose privacy metric is estimation variance, yields the lowest $\tau_k$ while CP yields slightly higher $\tau_k$ and IB and PF yield the highest $\tau_k$. In CP, the utility projection $\bGamma^{(\calP)}_k$ and privacy projection $\bGamma^{(\calQ)}_k$ cannot capture, respectively, the entire utility subspace $\left[\bH_k\right]_{:,\cal P}$ and the entire privacy subspace $\left[\bH_k\right]_{:,\cal Q}$, unless $\bH_k$ is an orthogonal matrix. \cref{fig:orth} shows the results when $\bH_k$ is an orthogonal matrix. We see that $\tau_k$ in our proposed scheme is always lower than that in CP while $\eta_k \geq |\calQ|\delta$ from $k=3$ onwards for both schemes.

\subsection{Decentralized sensors}
\label{sect:sim_dist}
In this subsection, we consider the case with multiple decentralized sensors. Let ${\bf Q}_k=4{\bf I}_L$, $r_k=r_{k,\star}$ in \cref{rkstar}, $\delta=3$ and $N=30$. We randomly generate each  $\bF_k$ for all $k\geq1$ so that its singular values are uniformly drawn from $[1,1.2]$, and use \cref{eq:eqk2} as the private error function. \cref{fig:sufficient} uses $S=3$ sensors and each sensor $s$'s measurement has dimension $N_s=N/S=10$, which is greater than the number of unknowns $L=8$, whereas \cref{fig:deficient} uses $S=5$ sensors with $N_s=N/S=6$, which is less than $L$ thus making each sensor an under-determined system. Comparing \cref{fig:sufficient} where $N_s=10$ and \cref{fig:deficient} where $N_s=6$, we notice that the compression matrix of $N_s=6$ is a submatrix of the compression matrix of $N_s=10$. Thus, $\tau_k$ for $N_s=10$ is less than that for $N_s=6$. 
In \cref{fig:sufficient,fig:deficient}, while the sequential scheme yields $\min(\eta_k)=\delta$ for all time steps,  the "no info. exch." scheme yields $\min(\eta_k)<\delta$ thus no privacy guarantee is ensured. 
In \cref{fig:deficient} where each sensor is an under-determined system, the estimation errors of both public and private states are infinitely large at each sensor. Therefore, the no information exchange scheme at each sensor $s$ will choose $M_s$ to be $N_s$ for all $s=1,\ldots, S$, i.e., no compression is used. As a result, after aggregating the measurements from all sensors at fusion center, the estimation errors of both public and private states are as low as the unsanitized ones. This explains why both $\min(\eta_k)$ and $\tau_k$ obtained using the no information exchange scheme are very small in \cref{fig:deficient}.


\begin{figure}[!htb]
\centering
    \includegraphics[width=.6\linewidth]{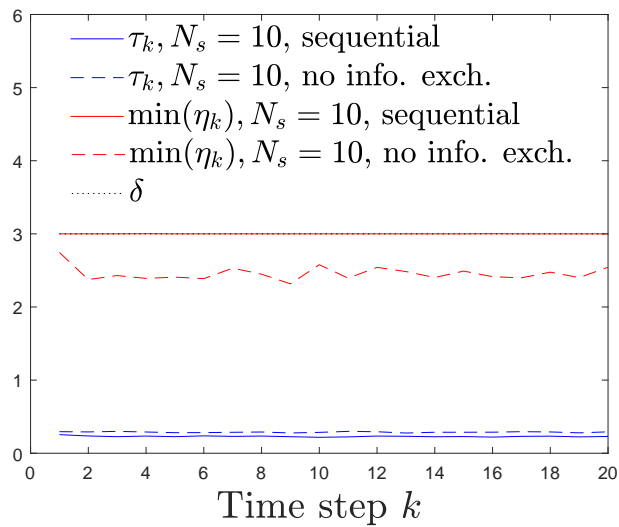}
\caption{The performance of the sequential scheme where $\text{Num-iter}=10$ and $N_1,\ldots,N_3=10$, in terms of public error trace $\tau_k$ and the minimum private error variance $\min(\eta_k)$.}\label{fig:sufficient}
\end{figure}

\begin{figure}[!htb]
\centering
    \includegraphics[width=.6\linewidth]{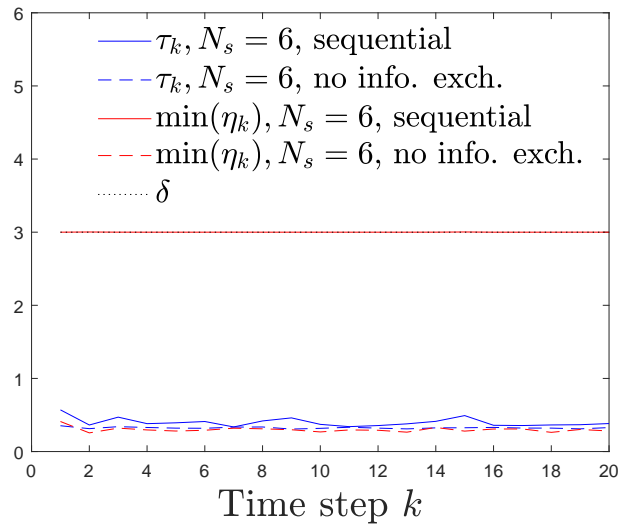}
\caption{The performance of the sequential scheme where $\text{Num-iter}=10$ and $N_1,\ldots,N_5=6$, in terms of public error trace $\tau_k$ and the minimum private error variance $\min(\eta_k)$.}\label{fig:deficient}
\end{figure}

\cref{fig:conv}, which uses the same settings as in \cref{fig:sufficient}, shows how $\tau_k$ and $\min(\eta_k(\bC_k))$  evolve over the time steps $k$ for the sequential scheme. Both $\tau_k$ and $\min(\eta_k(\bC_k))$ converge after a few iterations. 

\begin{figure}[!htb]
\centering
    \includegraphics[width=.6\linewidth]{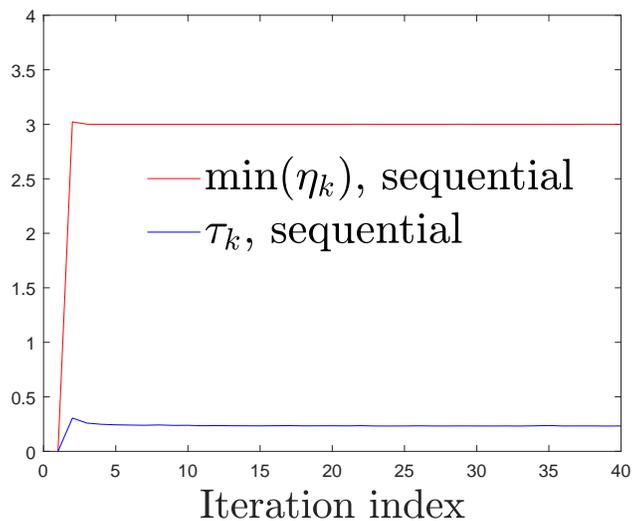}
  \caption{Convergence of $\tau_k$ and $\min(\eta_k(\bC_k))$ using \cref{alg:decentralized}, where $S=3$, $N_1,\ldots,N_3=10$ and $\delta=3$. The results shown are for time step $k=20$.}
  \label{fig:conv}
\end{figure}

\section{Experimental Results}
\label{sect:ekf}

\subsection{Privacy-aware localization}
We conduct an experiment using DecaWave UWB sensors \cite{dwm1001} for localization. We place 5 anchor nodes at known locations $\bp^{(A)}\in\Real^{2\times 5}$ and estimate a mobile node's 2-dimensional (2D) trajectory. The mobile node's state at time $k$ is denoted as $\bx_k = \left[v_{k}, \theta_k, p_{x,k}, p_{y,k}\right]\T$. The public state $\left[\bx_k\right]_{\cal P} = v_k$ is the mobile node's speed while the private state $\left[\bx_k\right]_{\cal Q} = \left[\theta_k, p_{x,k}, p_{y,k} \right]\T$ contains the node's heading and 2D location. The measurements are anchor-to-mobile ranges, which are a non-linear function of the state: $\bz_k = h_k(\bx_k)+\bn_k$. We use the extended Kalman filter (EKF) in our estimation procedure. At each time $k$, we linearize $h_k$ and $f_k$ around the estimate $\hat{\bx}_{k|k-1}$ and define the measurement matrix and the transition matrix, respectively, to be $\bH_k = \left. \frac{\partial h_k(\bx)}{\partial \bx}\right|_{\bx=\hat{\bx}_{k|k-1}}$ and $\bF_k = \left. \frac{\partial f_k(\bx)}{\partial \bx}\right|_{\bx=\hat{\bx}_{k-1|k-1}}$, where 
$f_k(\bx_{k})=\left[v_{k},\theta_{k},p_{x,k}+\Delta v_k\cos\theta_k,p_{y,k}+\Delta v_k\sin\theta_k
\right]\T$  with $\Delta=0.1$ being the sampling interval. The transition matrix  is approximated by
\begin{align*}
\bF_k=\left[
\begin{array}{cccc}
1&0&0&0\\
0&1&0&0\\
\Delta\cos\hat{\theta}_{k-1}& -\Delta\sin\hat{\theta}_{k-1} \hat{v}_{k-1}& 1& 0\\
\Delta\sin\hat{\theta}_{k-1}&  \Delta\cos\hat{\theta}_{k-1} \hat{v}_{k-1}& 1& 0
\end{array}
\right].
\end{align*}

The measurement $\bz_k\in \Real^{5\times 1}$ contains range measurements between the mobile node and anchor nodes with 
\begin{align*}
\left[\bz_k\right]_i = \left\|\left[\bp^{(A)}\right]_{:,i}-\left[\bx_k\right]_{\cal Q}\right\|+\left[\bn_k\right]_i,i=1,\ldots,5,
\end{align*}
where $\|\cdot\|$ denotes the Euclidean norm. Under the EKF framework, we let the measurement matrix $\bH_k\in\Real^{5\times 4}$ to be given by 
\begin{align*}
\left[\bH_k\right]_i = -\left[ 0,0,\frac{\left[\bp^{(A)}\right]_{1,i}-\left[\hat{\bx}_{k|k-1}\right]_3}{\left\|\left[\bp^{(A)}\right]_{:,i}-\left[\hat{\bx}_{k|k-1}\right]_{\cal Q}\right\|}, \frac{\left[\bp^{(A)}\right]_{2,i}-\left[\hat{\bx}_{k|k-1}\right]_4}{\left\|\left[\bp^{(A)}\right]_{:,i}-\left[\hat{\bx}_{k|k-1}\right]_{\cal Q}\right\|}\right], \ i=1,\ldots,5,
\end{align*}
where $\hat{\bx}_{k|k-1}$ is the predicted state before measurements at time $k$. Recall that we want to protect the mobile node's location while estimating its velocity, which however is not directly observed. In \cref{fig:ekf_pub} and \cref{fig:ekf_pri}, we consider the centralized setting where the mobile node collects the ranging measurements and sanitizes them before sending to a server. We set $\bP_0=0.01\bI_4$ and $\bR_k=0.04\bI_5,\bQ_k=0.04\bI_4$, and $r_k=0$ for all time steps $k$. The private error function used is \cref{eq:eqk1}. 

The mobile node's state is estimated over 318 time steps with or without compressive sanitization. 
We let $|\calQ|\delta=2$. 
The trajectory estimated with compressive sanitization (red crosses in \cref{fig:ekf_pri}) is noticeably distorted as compared to the ground true line (black line) thus the private information is protected whereas the private information is exposed in the trajectory estimated without compressive sanitization (blue dashes), which matches the ground truth line well. 
On the other hand, the public information (red line in \cref{fig:ekf_pub}), i.e., the speed of the mobile node matches fairly well with the unsanitized one (black line). This experiment shows the  proposed approach is also applicable to non-linear dynamical systems that can be approximated using the EKF framework.

\begin{figure}[!htb]
\centering
    \includegraphics[width=.8\linewidth]{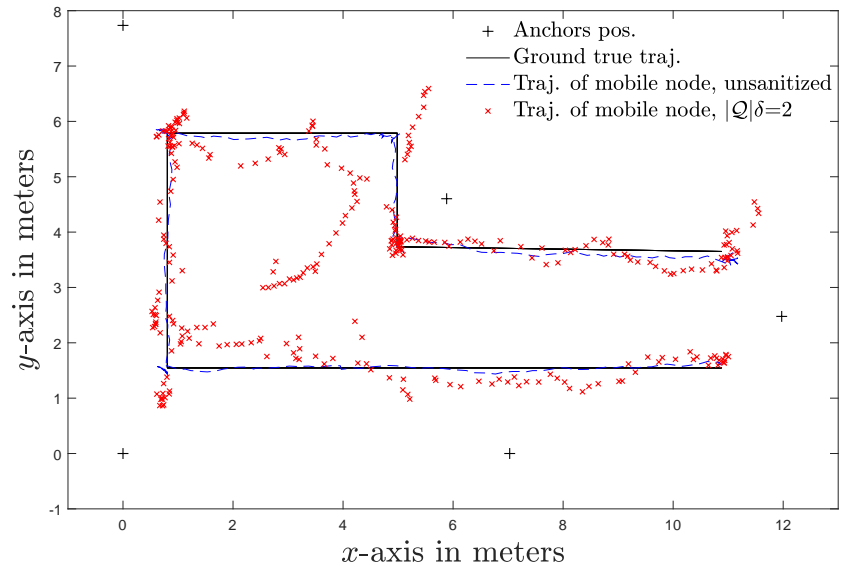}
\caption{Estimation of the mobile node's locations (private states).}
\label{fig:ekf_pri}
\end{figure}

\begin{figure}[!htb]
\centering
    \includegraphics[width=.6\linewidth]{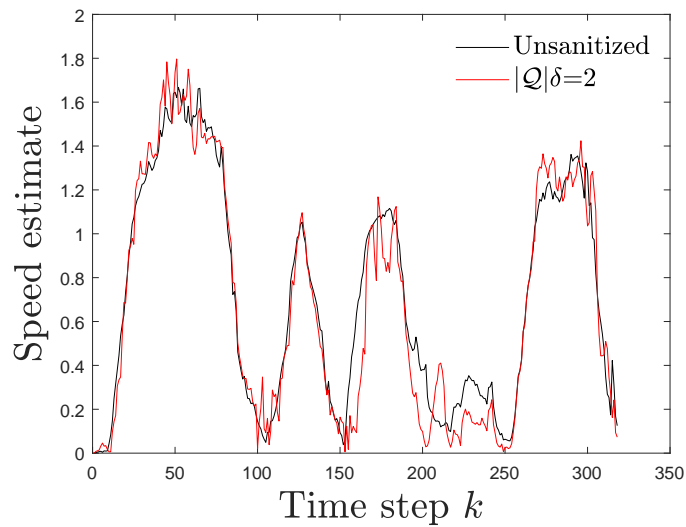}
\caption{Estimation of the mobile node's speed (public states).}
\label{fig:ekf_pub}
\end{figure}

\subsection{Privacy-aware human activity recognition}
In this section, we test our approach using the public data set for human activity recognition (HAR) from \cite{har}. We consider three activity labels: walking, walking upstairs and walking downstairs, in which walking is the public label while walking upstairs/downstairs are the private labels (e.g., these two activities can reveal how many floors the person's house has). We evaluate the impact of all 561 features HAR provides on the public or private activity labels using 10-fold cross-validation and rank them according to information gain. This process is automated using WEKA \cite{weka}. We choose five features ($|\calP|=5$) that have the biggest impact on walking and five features ($|\calQ|=5$) that have the biggest impact on walking upstairs/downstairs. The public features and private features are mutually exclusive. There are 30 subjects involved in HAR. For each subject, we update all $L=10$ features successively at each time step using our approach, where we set $\bP_{0|0}=0.01\bI_{10}$, $\bF_k=\bI_{10}, \bQ_k=0.2\bI_{10}$, and $\bH_k\in\Real^{N_k\times L}$ to be an identity matrix $\bI_L$ with its rows being randomly dropped with probability of 0.8 at each time step to simulate lossy transmission, and $\bR_k=0.1\bI_{N_k}$. \cref{fig:har} shows the public error trace $\tau_k$ and private error trace $\eta_k$ where we choose \cref{eq:eqk1} to be the private error function for \cref{fig:har} so that $|\calF|=1$ and $\calF(\delta\bm{1}_{|\calQ|})=|\calQ|\delta$ and $\delta=1$ and results are obtained by averaging over 30 subjects. We see that our framework is able to generate a low public error trace $\tau_k$ and a high private error trace $\eta_k$.
\begin{figure}[!htb]
\centering
    \includegraphics[width=.6\linewidth]{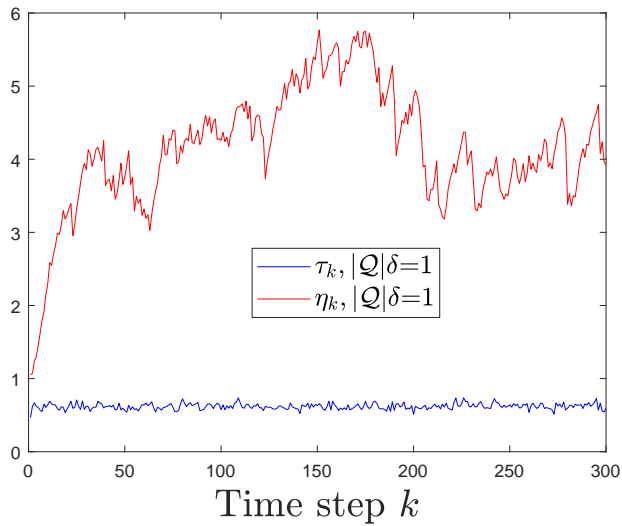}
\caption{The public error trace $\tau_k$ and private error trace $\eta_k$ for estimating public features and private features, respectively.}
\label{fig:har}
\end{figure}

\section{Conclusion}
\label{section:c}
In this paper, we have investigated the use of compressive privacy schemes in a LDS to prevent the fusion center from estimating a set of private states accurately while still allowing it to estimate a set of public states with good accuracy. We developed an optimization framework to find the optimal compression matrix at each time step to achieve an optimal tradeoff between the utility at the current time step and privacy protection at multiple steps in the future. We showed that this approach allows us to ensure the same level of privacy in all future time steps in a LDS. We develop algorithms to solve for the optimal compression matrix in both the centralized and decentralized settings. For the decentralized case, we proposed the sequential update schedule. Extensive simulations are performed to verify the performance of our proposed approaches with comparisons to other methods in the literature, which are however not designed for LDS. Two empirical experiments are conducted to verify that our proposed approach on real-world time series data.

In this paper, we have not considered the case where side information at each time step may be used in privacy attacks to infer the private states. In our formulation, such \emph{initial} side information can be captured in the prior distribution at time step 0. It is of interest in future research work to consider the availability of side information at every time step.
Moreover, our approach is model dependent and thus accurate measurement and state models are essential for our approach. 

\appendices


\section{Proof of Lemma \ref{lemma:U=A}}\label{app:proof_lemma:U=A}

Let $\bA'=\bA\bB^{1/2}$. Let ${\bf V} {\bf \Lambda} {\bf U}\T$ be the truncated version of the singular value decomposition (SVD) of $\bA'$, where ${\bf \Lambda}\in\Real^{M\times M}$ is a diagonal matrix with non-zero singular values being its diagonal elements, ${\bf V}\in\Real^{M\times M}$ contains the $M$ left singular vectors, and $\bU$ is as defined in the lemma statement. Substituting the SVD of $\bA'$ into the left-hand side of \cref{eq:ABA}, we obtain the right-hand side. This completes the proof.

\section{Proof of Lemma \ref{lemma:rank}}\label{app:lemma:rank}

Let $\{\lambda_{\bA,i},i=1,\ldots,N\}$ and $\{\lambda_{-\bB,i},i=1,\ldots,N\}$ denote the eigenvalues of $\bA$ and $-\bB$, respectively, sorted in descending order. The eigenvalues of ${\bf C}$, denoted as $\{\lambda_{\bC,i},i=1,\ldots,N\}$ sorted in descending order, satisfy Weyl's inequality:
\begin{align*}
\lambda_{\bA,N}+\lambda_{-\bB,i}
\leq
\lambda_{C,i}
\leq 
\lambda_{\bA,i}+\lambda_{-\bB,1}, i=1,\ldots,N.
\end{align*}
Since $N>r_\bA+r_\bB$, $\lambda_{\bA,i}=0$ for $i>r_{\bA}$ and $\lambda_{-\bB,i}=0$ for $i\leq N-r_\bB$. Then, we have
\begin{align*}
0\leq \lambda_{\bC,i} \leq \lambda_{\bA,i}, &\ 1\leq i \leq r_\bA, \\
\lambda_{\bC,i}=0, &\ r_\bA+1\leq i \leq N-r_\bB \\
\lambda_{- \bB,i} \leq \lambda_{\bC,i} \leq  0,&\ N-r_\bB+1 \leq i \leq N.
\end{align*}
The proof is now complete.

\section{Derivation of (\ref{eq:pkplusn_info_ex_us})}
\label{A2}
Let ${\bPi}_{i,j}$ be a permutation matrix that swaps the $i$th and the $j$th columns of the identity matrix ${\bf I}_S$. Then, right multiplying a matrix
by ${\bPi}_{s\rightarrow 1}={\bPi}_{s,s-1}{\bPi}_{s-1,s-2}\cdots{\bPi}_{2,1}$ ends up moving the $s$th column of this matrix to its $1$st column and it's easy to verify that ${\bPi}_{s\rightarrow 1}{\bPi}_{s\rightarrow 1}\T=\bI_s$. In what follows, we firstly show that $\bD_{k+n|k}$ does not depend on the order of sensors and then separate the terms depending on $\pC_\JsIs$ in $\bD_{k+n|k}$ from the terms depending on $\pC_\JnsIns$. From \cref{eq:p_k_tilde}, we have
\begin{align}
\bD_{k+n|k} =&
\left[\pG_{\calI_1}\T\pC\T_{\calJ_1,\calI_1},\ldots,\pG_{\calI_S}\T\pC\T_\JsIs\right]
{\bPi}_{s\rightarrow 1} \nonumber\\
&{\bPi}_{s\rightarrow 1}\T
{\bPsi_k }^{-1}
{\bPi}_{s\rightarrow 1}
{\bPi}_{s\rightarrow 1}\T
\left[\begin{array}{c}
\pC_\JlIl \pG_{\calI_1} \\
\vdots\\
\pC_\JsIs \pG_{\calI_S}
\end{array}
\right],
\nonumber \\
=&
\left[\pG_{\calI_1}\T\pC\T_\JlIl,\ldots,\pG_{\calI_S}\T\pC\T_\JsIs\right]
{\bPi}_{s\rightarrow 1} \nonumber\\
&\left(
{\bPi}_{s\rightarrow 1}\T
{\bPsi_k }
{\bPi}_{s\rightarrow 1}
\right)^{-1}
{\bPi}_{s\rightarrow 1}\T
\left[\begin{array}{c}
\pC_\JlIl\pG_{\calI_1} \\
\vdots\\
\pC_\JsIs\pG_{\calI_S}
\end{array}
\right],
\nonumber\\
=&
\left[\pG\T_{\calI_s}\pC\T_\JsIs, \pG\T_{\calI_{\s}}\pC\T_\JnsIns\right] 
 \nonumber\\
& 
\left[\begin{array}{cc}
\pC_\JsIs\pT_\IsIs\pC\T_\JsIs & 
\pC_\JsIs\pT_{\IsIns}\pC\T_\JnsIns \\
\pC_\JnsIns\pT_{\IsIns}\T\pC\T_\JsIs & 
[\bPhi_k]_{\InsIns}
\end{array}
\right]^{-1}  \nonumber\\
&\left[\begin{array}{c}
\pC_\JsIs \pG_{\calI_s} \\
\pC_\JnsIns \pG_{\calI_{\s}}
\end{array}
\right], \nonumber\\
=&
\pG_{\calI_s}\T\pC\T_\JsIs \left(\pC_\JsIs \bT_k^{(s)} \pC\T_\JsIs\right)^{-1} \nonumber\\
&\pC_\JsIs\pG_{\calI_s} \nonumber\\
&-
\pG_{\calI_s}\T\pC\T_\JsIs
\left(\pC_\JsIs  \bT_k^{(s)}  \pC\T_\JsIs\right)^{-1} \nonumber\\
&~~~
\pC_\JsIs \pT_{\IsIns}\pC\T_\JnsIns [\bPhi_k]^{-1}_{\InsIns} \pC_\JnsIns\pG_{\calI_{\s}}  \nonumber\\
&-
\pG_{\calI_{\s}}\T\pC\T_\JnsIns [\bPhi_k]^{-1}_{\InsIns} \pC_\JnsIns\pT_{\IsIns}\T \nonumber\\
&\pC\T_\JsIs
 \left(\pC_\JsIs \bT_k^{(s)} \pC\T_\JsIs\right)^{-1} 
\pC_\JsIs\pG_{\calI_s}  \nonumber\\
&+
\pG_{\calI_{\s}}\T  \pC\T_\JnsIns [\bPhi_k]^{-1}_{\InsIns}  \pC_\JnsIns\pT_{\IsIns}\T \nonumber\\
&\pC\T_\JsIs\left(\pC_\JsIs \bT_k^{(s)} \pC\T_\JsIs\right)^{-1}
\pC_\JsIs \pT_{\IsIns}\nonumber\\ 
&\pC_\JnsIns\T [\bPhi_k]^{-1}_{\InsIns} \pC_\JnsIns \pG_{\calI_{\s}} \nonumber\\
&+ 
\pG_{\calI_{\s}}\T \pC\T_\JnsIns [\bPhi_k]^{-1}_{\InsIns} \pC_\JnsIns \pG_{\calI_{\s}}, \nonumber
\end{align} 
where the last equality follows from the matrix inversion lemma. Applying \cref{lemma:U=A}, we obtain \cref{eq:pkplusn_info_ex_us}.

%


\bibliographystyle{IEEEtran}
\bibliography{IEEEabrv,StringDefinitions,privacy}

\end{document}